\documentclass[letter,compsoc]{IEEEtran}
\usepackage[T1]{fontenc}
\usepackage{bm}
\usepackage{amsmath}
\usepackage{amsthm}
\usepackage{amssymb}
\usepackage{graphicx}
\usepackage{tabularx, makecell, multirow}
\usepackage{makecell}
\usepackage{stfloats}
\usepackage{booktabs}
\usepackage{multicol}
\usepackage{multirow}
\usepackage{verbatim}
\usepackage{makecell}
\usepackage{color}
\usepackage{tabularx}
\usepackage{subfig}
\usepackage{float}  
\DeclareSubrefFormat{parens}{#1 (#2)}

\usepackage[unicode=true,
bookmarks=true,bookmarksnumbered=true,bookmarksopen=true,bookmarksopenlevel=1,
breaklinks=false,pdfborder={0 0 0},pdfborderstyle={},backref=false,colorlinks=false]
{hyperref}
\hypersetup{pdftitle={Your Title},
	pdfauthor={Your Name},
	pdfpagelayout=OneColumn, pdfnewwindow=true, pdfstartview=XYZ, plainpages=false}

\makeatletter
\theoremstyle{plain}
\newtheorem{thm}{\protect\theoremname}
\theoremstyle{definition}

\theoremstyle{remark}

\theoremstyle{plain}

\theoremstyle{plain}

\usepackage[vlined,boxed,ruled,linesnumbered,resetcount]{algorithm2e}
\usepackage{algpseudocode}
\usepackage[noadjust,sort,compress]{cite}

\providecommand{\definitionname}{Definition}
\providecommand{\lemmaname}{Lemma}
\providecommand{\propositionname}{Proposition}
\providecommand{\remarkname}{Remark}
\providecommand{\theoremname}{Theorem}
\usepackage[justification=centering]{caption}

\newcommand {\aplt} {\ {\raise-.5ex\hbox{$\buildrel<\over{\mbox{\scriptsize $\sim$}}$}}\ }

\usepackage{color, colortbl}
	
\definecolor{Gray}{gray}{0.9}
\usepackage[first=0,last=9]{lcg}

\makeatother
 
\begin{document}
\title{FedReverse: Multiparty Reversible Deep Neural Network Watermarking}
 
\author{Junlong Mao,  Huiyi Tang, Yi Zhang, Fengxia Liu, Zhiyong Zheng and Shanxiang Lyu
	\thanks{
	Junlong Mao,  Huiyi Tang and Shanxiang Lyu are with the College of Cyber Security, Jinan
		University, Guangzhou 510632, China (Emails: maojunlong@stu2022.jnu.edu.cn,
		wt20180112@stu2019.jnu.edu.cn, lsx07@jnu.edu.cn). Yi Zhang, Fengxia Liu and Zhiyong Zheng are with the Engineering Research Center of Ministry of Education for Financial Computing and Digital Engineering, Renmin University of China, Beijing 100872, China (Emails: ethanzhang@ruc.edu.cn, shunliliu@buaa.edu.cn, zhengzy@ruc.edu.cn). 
	}
}

\maketitle

\begin{abstract}
	The proliferation of Deep Neural Networks (DNN) in commercial applications is expanding rapidly. Simultaneously, the increasing complexity and cost of training DNN models have intensified the urgency surrounding the protection of intellectual property associated with these trained models.
In this regard, DNN watermarking has emerged as a crucial safeguarding technique. This paper presents FedReverse, a novel multiparty reversible watermarking approach for robust copyright protection while minimizing performance impact. Unlike existing methods, FedReverse enables collaborative watermark embedding from multiple parties after model training, ensuring individual copyright claims. In addition, FedReverse is reversible, enabling complete watermark removal with unanimous client consent.  FedReverse demonstrates perfect covering, ensuring that observations of watermarked content do not reveal any information about the hidden watermark. Additionally, it showcases resistance against Known Original Attacks (KOA), making it highly challenging for attackers to forge watermarks or infer the key. This paper further evaluates FedReverse through comprehensive simulations involving Multi-layer Perceptron (MLP) and Convolutional Neural Networks (CNN) trained on the MNIST dataset. The simulations demonstrate FedReverse's robustness, reversibility, and minimal impact on model accuracy across varying embedding parameters and multiple client scenarios.
\end{abstract}


\section{Introduction}\label{Sec. introduction}
The soaring popularity of Deep Neural Networks (DNN) can be attributed to their outstanding performance in various domains \cite{lecun2015deep,kamilaris2018deep, guo2016deep, samek2021explaining,tang2021steganography}. However, the widespread adoption of DNNs has raised concerns regarding unauthorized model usage and a lack of proper attribution to their creators \cite{DBLP:conf/uss/JagielskiCBKP20,DBLP:conf/crypto/CarliniJM20,DBLP:conf/ccs/JuutiAA19}. In response to these challenges, the field of DNN watermarking has emerged as a vital means of safeguarding the intellectual property embedded within these models \cite{barni2021challenge}. Watermarking offers an additional layer of security that enables creators to assert ownership, defend models against unauthorized access and tampering, trace their origins, ensure data integrity, manage versions, and detect malicious usage \cite{DBLP:conf/uss/AdiBCPK18,regazzoni2021protecting,DBLP:conf/iclr/KrishnaTPPI20,DBLP:conf/asplos/RouhaniCK19}.

To protect the intellectual property rights of  DNNs, a range of DNN watermarking techniques have been developed, including \textit{parameter-based} watermarking and \textit{backdoor-based} watermarking, which are discussed in references \cite{li2021survey,zhang2018protecting}.
i) Parameter-based watermarking methods involve the embedding of personalized watermarks into the parameters or their distribution within the DNN model, as elaborated in references \cite{uchida2017embedding,regazzoni2021protecting,zhang2018protecting,1227616}, and \cite{qin2022lattice}. While this approach allows for the conveyance of multiple bits of information, it requires access to the inner workings of the suspected model, known as white-box access, for watermark extraction.
ii) Backdoor-based watermarking methods, as elucidated in references \cite{li2021survey,gong2020privacy,lyu2023optimized,qin2023reversible}, exploit backdoor attacks \cite{li2016multiple} to introduce specific triggers that can identify the model's ownership. However, backdoor-based watermarking typically results in a zero-bit watermark, which means that it only indicates the presence or absence of the watermark itself, rather than representing an identity in the form of a bit string \cite{DBLP:journals/corr/abs-2208-14127}. This verification can be achieved even with limited information, known as black-box access.

Recent research has also explored the use of  watermarks to protect copyright in the context of Federated Learning (FL) \cite{gong2020privacy, zhang2021survey,DBLP:journals/titb/HanJWQD23,chen2023fedright}. WAFFLE \cite{tekgul2021waffle} seems to be the first DNN watermarking technique for FL. It assigns the Server the responsibility of integrating backdoor-based watermarks into the FL model. Clients cannot backdoor, poison, or embed their own watermarks since they are incentivized to maximize global accuracy.  FedIPR \cite{li2022fedipr} and  \cite{liu2021secure,yang2022watermarking}  advocate for the hybrid of  black-box and white-box techniques. They enable all
clients to embed their own watermark in the global model  without sharing secret information. On this basis,  FedTracker \cite{shao2022fedtracker} has been proposed to provide the
trace evidence for illegal model re-distribution by unruly clients.
It uses a combination of server-side global backdoor watermarks
and client-specific local parameter watermarks, where the former is
used to verify the ownership of the global model, and the latter is
used to trace illegally re-distributed models back to unruly clients.
  Nevertheless, existing studies primarily focus on embedding watermarks during training, which could potentially diminish the system's performance. Another crucial challenge is ensuring that private watermarks added by different clients to the same federated DNN model do not conflict with each other. This challenge is unique to the federated learning setting, where different client's watermarks may have the potential to undermine one another.

To address the aforementioned challenges, we introduce a multiparty reversible watermarking scheme, referred to as FedReverse. It exhibits the following distinct features:
\begin{itemize}
 \item \textit{Reversibility}: Unlike conventional methods that embed watermarks during training, FedReverse uniquely incorporates client watermarks into the model's weights post-training. This distinctive approach empowers individual clients to assert exclusive ownership rights over the trained model. The reversible nature of FedReverse allows for the complete removal of watermarks without compromising the model's original weights, ensuring fidelity and flexibility in watermark management.
	\item \textit{Multiparty}: FedReverse enables the trained model to acknowledge the credits of multiple clients. It mitigates potential watermark conflicts among diverse clients by employing an orthogonal key generation technique. This innovative method assigns each client a unique key aligned with a vector in a random matrix. By projecting the cover vector $\mathbf{s}$ onto the direction defined by each key, FedReverse employs a lattice-based reversible data hiding technique within the projected space.  To fully restore the model's weights to their original state, our scheme requires the unanimous consent of all involved parties.
	\item \textit{Security and Reliability}: In terms of safeguarding intellectual property associated with trained DNN models, FedReverse demonstrates robust resistance against Known Original Attacks (KOA), significantly challenging potential attackers in forging watermarks or deducing the secret key. This paper also extensively evaluates FedReverse through comprehensive simulations involving Multi-layer Perceptron (MLP) and Convolutional Neural Networks (CNN) trained on the MNIST dataset, showcasing FedReverse's robustness, reversibility, and its negligible impact on model accuracy across varying embedding parameters and diverse client scenarios. 
\end{itemize}

\noindent Notations: Vectors are represented by lowercase boldface letters. $|\cdot|$, $\langle \cdot, \cdot \rangle$ and $\left\|\cdot \right\| $ respectively denote the element-wise absolute value, the inner product, and the Euclidean norm of the input. The projection operator is defined as $\text{proj}_\mathbf{u}(\mathbf{s}) = \frac{\langle\mathbf{s},\mathbf{u}\rangle}{\|\mathbf{u}\|^2} \cdot \mathbf{u}$.  And the subscript of a parenthesized vector represents the corresponding element of the vector, like $(\mathbf{u}_i)_j$ as $j$-th element of $\mathbf{u}_i$. 
%

\section{Preliminaries}\label{Sec. related work}
    \subsection{Problem Formulation}
    Reversible DNN watermarking embeds watermarks into the weights of neural networks after the model has been trained. Let $\mathcal{W}$ denote the set of all weights in a trained DNN model. During watermark embedding, specific weights from $\mathcal{W}$ are selected based on a location sequence to generate a cover sequence $\mathbf{s}$.
    In the conventional one-party watermarking, aided by the secret key $K$, the embedding, watermark extraction, and weight recovery are given by the following triplet of operations
    \begin{equation}
    	\left\{ \begin{aligned}
    	\mathbf{y} &= \mathrm{Emb}(\mathbf{s}, \mathbf{m},K) \\
    		\hat{\mathbf{m}} &= \mathrm{Ext}(\mathbf{y},K) \\
    		\hat{\mathbf{s}} &=   \mathrm{Rec}(\mathbf{y},K) \\
    	\end{aligned}\right.
    \end{equation}
    where  $\mathrm{Emb}(\cdot)$  embeds the information sequence $\mathbf{m}$ into the cover sequence $\mathbf{s}$ to produce the watermarked sequence $\mathbf{y}$, $\mathrm{Ext}(\cdot)$ and $\mathrm{Rec}(\cdot)$ denote the extraction and recovery functions, respectively.  A reversible watermarking scheme featureing $\mathrm{Rec}(\cdot)$ enables $\hat{\mathbf{s}}=\mathbf{s}$, which differs from non-reversible watermarking.
    
    To enable multiparty watermarking for $n$ clients, the embedding, extraction and recovery triplet is formulated as
        \begin{equation}
    	\left\{ \begin{aligned}
    	\mathbf{y} &= \mathrm{Emb}(\mathbf{s}, \mathbf{m}_1, \ldots, \mathbf{m}_n, K_1, \ldots, K_n) \\
    		\hat{\mathbf{m}}_i &= \mathrm{Ext}(\mathbf{y},K_i), \quad i=1, \ldots, n \\
    		\hat{\mathbf{s}} &=   \mathrm{Rec}(\mathbf{y},K_1, \ldots, K_n)  \\
    	\end{aligned}\right.
    \end{equation}
Here the watermarks $\mathbf{m}_1, \ldots, \mathbf{m}_n$ are embedded simultaneously into the cover sequence $\mathbf{s}$. The extraction is performed individually, so each of the client can claim his/her copyright. In addition, all the clients should cooperate to recover  $\hat{\mathbf{s}}$.
 
\subsection{Reversible Watermarking by Difference Contraction}\label{Difference Contraction}
\textit{Difference expansion} \cite{1227616} is the crux for most reversible data hiding techniques, which is feasible for integer $\mathbf{s}$ with high correlations (e.g., a set of pixels from an image).
On the contrary, a lattice-based reversible data hiding employs the rationale of \textit{difference contraction} \cite{qin2022lattice}. This paradigm is more suitable for cover objects in the form of floating-point numbers.

Difference contraction can be summarized as follows. Let $Q_{\mathbf{m},K}$ be a quantization function identified by the information sequence $\mathbf{m}$ and secret key $K$ which serves the role of dithering \cite{zamir2014lattice}. While the quantization index modulation (QIM) method  employs $\mathbf{y}=Q_{\mathbf{m},K}(\mathbf{s})$ as the embedding function, the difference contraction method in \cite{qin2022lattice} employs 
\begin{align}
	\mathbf{y} &= \mathrm{Emb}_{\mathrm{DC}}(\mathbf{s}, \mathbf{m},K) \nonumber\\ &=Q_{\mathbf{m},K}(\mathbf{s})+ (1-\alpha) (\mathbf{s}-Q_{\mathbf{m},K}(\mathbf{s})). \label{DC-embed}
\end{align}
In the above, the difference vector $\mathbf{s}-Q_{\mathbf{m},K}(\mathbf{s})$ has been contracted by a factor of $1-\alpha$. The term $(1-\alpha) (\mathbf{s}-Q_{\mathbf{m},K}(\mathbf{s}))$ is regarded as a beneficial noise, which helps to achieve the reversibility of $\mathbf{s}$.

In terms of extraction, the receiver searches for the closest coset $
\Lambda_{\mathbf{m}}$ to $\mathbf{y}$ to extract the estimated message $\hat{\mathbf{m}}$ by
\begin{align}
	\hat{\mathbf{m}} &= \mathrm{Ext}_{\mathrm{DC}}(\mathbf{y},K) \\
&= \mathop{\arg\min}_{\mathbf{m}} \text{dist} (\mathbf{y}, \Lambda_\mathbf{m}), \label{DC-extract}
\end{align}
where $\mathrm{dist}(\mathbf{y}, \Lambda_{\mathbf{m}}) \triangleq \min_{\mathbf{v}\in \Lambda_{\mathbf{m}} }\|\mathbf{y}-\mathbf{v}\|$. If $Q_{\mathbf{m},K}(\mathbf{s}) = Q_{\hat{\mathbf{m}},K}(\mathbf{y})$, then the cover vector can be faithfully recovered by 
\begin{align}
	\hat{\mathbf{s}} &=   \mathrm{Rec}_{\mathrm{DC}}(\mathbf{y},K) \\
	&= \frac{\mathbf{y}- \alpha Q_{\hat{\mathbf{m}},K}(\mathbf{y})}{1-\alpha}. \label{DC-recover}
\end{align} 

In addition to the recovery of $\mathbf{s}$, its embedding distortion has also been reduced. Let $r$ be the size of $\mathbf{y}$ and $\mathbf{s}$, the Mean square error (MSE) is defined as 
\begin{equation}
	\mathrm{MSE}= \frac{1}{r} \mathbb{E}\left\| \mathbf{y}-\mathbf{s} \right\|^2.
\end{equation}
By employing $r$-dimensional integer lattices $\Delta\times \mathbb{Z}^r$ to defined the quantization function $Q_{\mathbf{m},K}$, and to set the amount of embedding information as $b$ bits per dimension, the MSE of the difference contraction method is
\begin{equation}
    \label{MSE_DC}
	\mathrm{MSE}_\mathrm{DC} = \frac{ \alpha^2 2^{2b}\Delta^2}{12}.
\end{equation}
 Moreover, the contraction factor $1-\alpha$ should satisfy $1-\alpha \leq 1/2^b$. It is noteworthy that $\Delta$ and $\alpha$ can be identified as public parameters, which controls the trade-off between embedding distortion and the robustness to additive noises.

\begin{figure*}[t]
	\centering
	\includegraphics[width=0.9\textwidth]{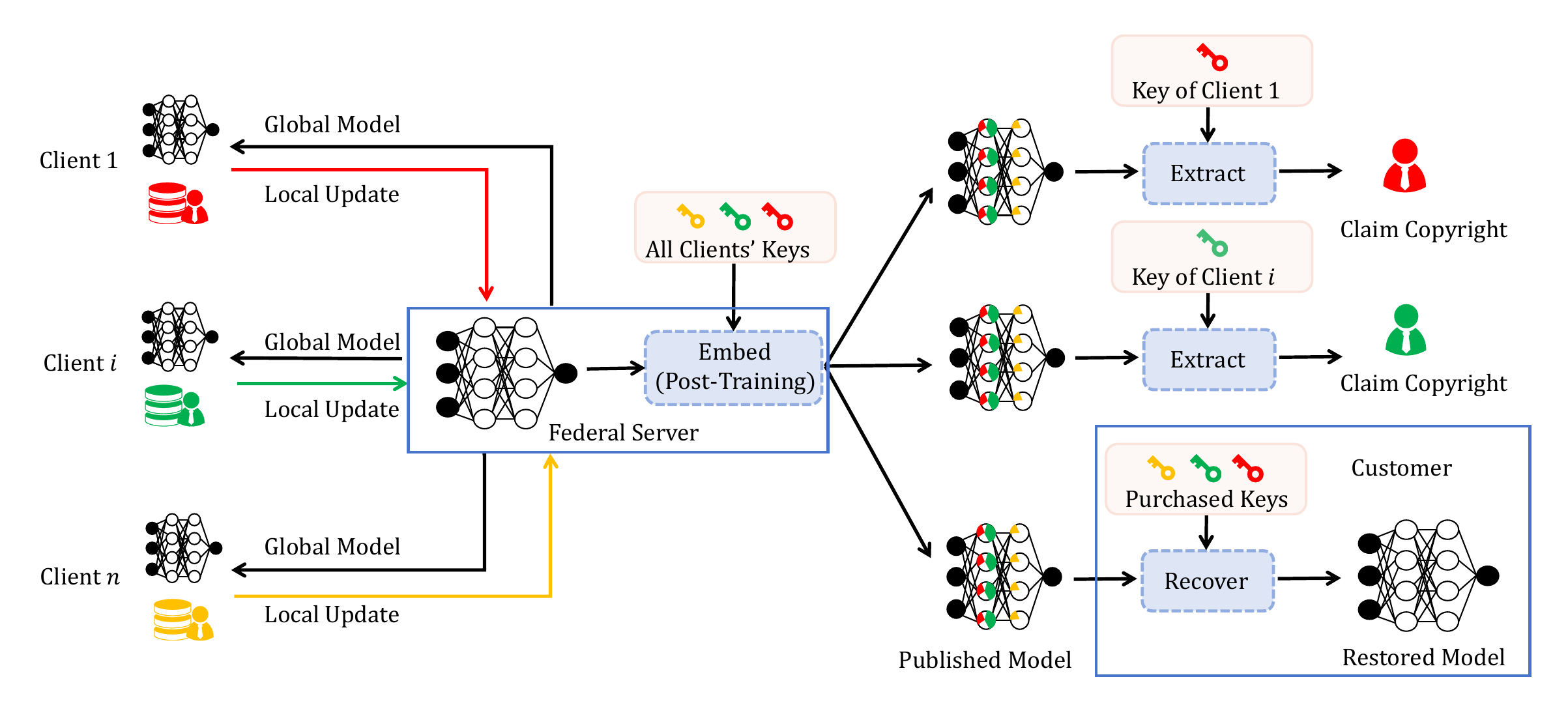}
	\caption{The schematic of multiparty reversible DNN watermarking.}
	\label{fig:proceeding}
\end{figure*}

\begin{figure}[t]
    	\centering
    	\includegraphics[width=0.49\textwidth]{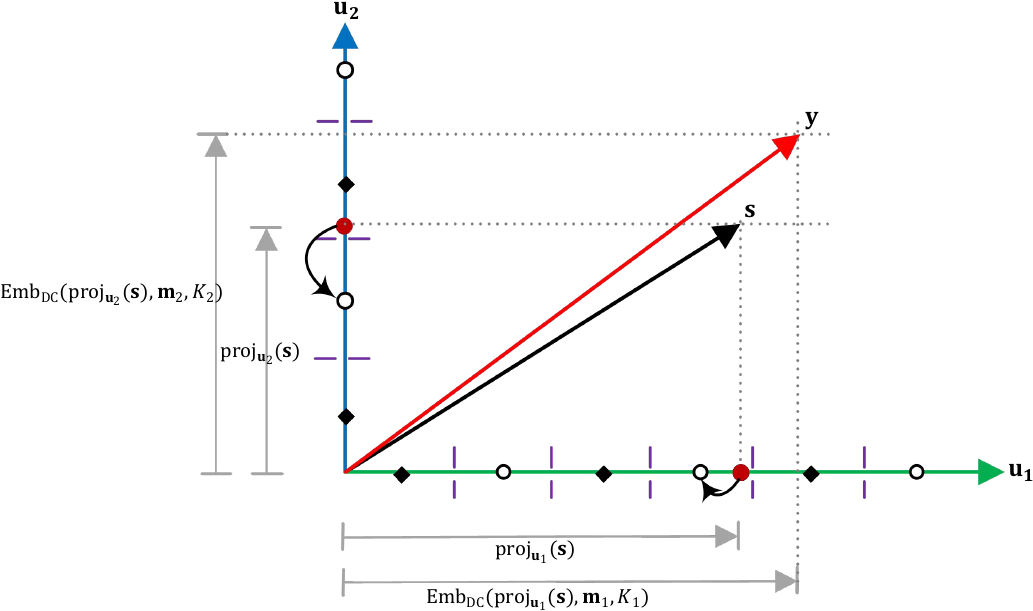}
    	\caption{Demonstration of embedding with two clients.}
    	\label{s_sw_fig1}
\end{figure}
 
\section{The Proposed Method: FedReverse}\label{SEC_PRO}
Expanding upon previously introduced concepts, FedReverse extends the principle of difference contraction to multiparty watermarking schemes. The schematic representation of FedReverse is depicted in Fig. \ref{fig:proceeding}. The primary elements are as follows: (1) During the training phase, FedReverse operates akin to conventional federated learning without integrating watermarks, thereby enhancing the accuracy of the trained model. (2) All clients engage in private key negotiation with the federated server. Subsequent to the training phase, the federated server utilizes these keys to embed reversible watermarks for all clients. (3) Concerning the published watermarked model, any client can assert their individual copyright over the model. (4) Prospective customers can obtain keys from all clients to recover a high-accuracy DNN model free of watermarks.

\subsection{Embedding, Extraction and Recovery}
\label{Multiparty Reversible Watermarking}
 
Each client, equipped with a unique key $K_i= \left\lbrace \mathbf{u}_i, d_i \right\rbrace$, aims to embed a message $\mathbf{m}_i \in \{0, 1\}$. The overall embedding function aggregates these embedded messages into the watermarked vector $\mathbf{y}$:
\begin{align}
	\mathbf{y} &= \mathrm{Emb}_{\mathrm{Fed}}(\mathbf{s}, \mathbf{m}_1, \ldots, \mathbf{m}_n, K_1, \ldots, K_n) \\
	&= \sum_{i=1}^{n} \mathrm{Emb}_{\mathrm{DC}}(\text{proj}_{\mathbf{u}_i}(\mathbf{s}) , \mathbf{m}_i,d_i) \nonumber \\
	&\quad + \mathbf{s} - \sum_{i=1}^{n} \text{proj}_{\mathbf{u}_i}(\mathbf{s}).
	\label{multiple_embed}
\end{align}
It's noteworthy that $\mathbf{s} - \sum_{i=1}^{n} \text{proj}_{\mathbf{u}_i}(\mathbf{s})=0$ only when $n=r$.

Regarding extraction, any client with key $K_i= \left\lbrace \mathbf{u}_i, d_i \right\rbrace$ can independently extract their embedded message from the received signal $\hat{\mathbf{y}}$ as follows:
\begin{align}
	\hat{\mathbf{m}}_i &= \mathrm{Ext}_{\mathrm{Fed}}(\mathbf{y},K_i) \\
	&= \mathrm{Ext}_{\mathrm{DC}}(\text{proj}_{\mathbf{u}_i}(\mathbf{y}),d_i) \nonumber \\
	&= \mathop{\arg\min}_{\mathbf{m}_i} \text{dist} (\text{proj}_{\mathbf{u}_i}(\mathbf{y}), \Lambda_{\mathbf{m}_i}).
	\label{multiple_extract}
\end{align}

Regarding recovery, the original signal $\hat{\mathbf{s}}$ can be fully restored using all client keys when the received signal remains undisturbed:
\begin{align}
	\hat{\mathbf{s}} &= \mathrm{Rec}_{\mathrm{Fed}}(\mathbf{y}, K_1, \ldots, K_n) \\
	&= \sum_{i=1}^{n} \mathrm{Rec}_{\mathrm{DC}}(\mathbf{y}, d_i) + \mathbf{y} - \sum_{i=1}^{n} \text{proj}_{\mathbf{u}_i}(\mathbf{y}). 
\end{align}

Fig. \ref{s_sw_fig1} illustrates the embedding process with two clients, where messages are embedded in their respective directions, and the watermarked projections are eventually merged into a single watermarked vector $\mathbf{y}$.

\subsection{Orthogonal Key Generation}
\label{Orthogonal Key Generation}

In FedReverse, the central server generates a set of orthogonal vectors $\left\lbrace \mathbf{u}_1, \ldots, \mathbf{u}_n \right\rbrace$ and one-dimensional dithers $\left\lbrace d_1, \ldots, d_n \right\rbrace$ as private keys $K_1, \ldots, K_n$. Obviously the dithers can be generated by random number seeds. We focus on the generation of $\left\lbrace \mathbf{u}_1, \ldots, \mathbf{u}_n \right\rbrace$  hereby.

\begin{enumerate}
	\item Each client chooses $n_i (n_i \geq 1)$ random numbers and sends them to the server, ensuring the total dimension $r = \sum n_i$.
	\item Using the received numbers as a seed, the server generates a random bit string and converts it to a matrix of size $r \times r$, as detailed in Algorithm~\ref{Random matrix generation algorithm}.
	\item The server orthogonalizes the matrix rows via Schmidt orthogonalization and sends $n_i$ row vector(s) to the $i$-th client.
	\item Clients receiving multiple vectors ($\mathbf{u}_i$) can generate new partial keys based on the received vectors by applying Algorithm~\ref{New partial key generation algorithm}.
\end{enumerate}

\begin{algorithm}[]
	\KwData{number of bits in a set $B$, vector dimension $r$, number of clients $n$, matrix entry range $q$, random number seed $seed$}
	\KwResult{A random matrix $\mathsf{KeyMat}$}
	$\mathsf{Key}$ $\leftarrow$ Random($seed$); 
	
	$\mathsf{KeyBin}$ $\leftarrow$ Dec2Bin($\mathsf{Key}$); \ // Convert to binary
	
	$\mathsf{KeyMat}[0, \cdots, r]$ $\leftarrow$ []
	
	Padding($\mathsf{KeyBin}$, $B \cdot r \cdot r$);  \ // Fill the length of $\mathsf{KeyBin}$ to $B \cdot r \cdot r$.
	
	\For{$j\leftarrow 0$ \KwTo $r$}{
		col[0, $\cdots$, r] $\leftarrow$ 0; 
		
		count $\leftarrow$ 0;
		
		\For{$i\leftarrow j \times r \times B$ \KwTo $(j+1) \times r \times B$}{
			col[count] $\leftarrow$ Bin2Dec(Kenbin[$i,\cdots,i+B$]) $\times$ ($q/(2^B)$)
			
			count $\leftarrow$ count + 1}
		$\mathsf{KeyMat}[j] \leftarrow $ col}
	
	\Return $\mathsf{KeyMat}$
	\caption{Random matrix generation algorithm}
	\label{Random matrix generation algorithm}
\end{algorithm}

\begin{algorithm}[]
	\KwData{Partial keys of a client $\mathsf{KeyU}$}
	\KwResult{A new partial key of the client $\mathsf{KeyNew}$}
	$\mathsf{KeyNew} \leftarrow \mathbf{0}$ 
	
	\For{$\mathbf{u} \ \mathrm{in} \ \mathsf{KeyU}$}{
		coef $\leftarrow$ Random()
		
		$\mathsf{KeyNew} \leftarrow \mathsf{KeyNew} + \mathrm{coef} \cdot \mathbf{u}$
	}
	
	\Return $\mathsf{KeyNew}$
	\caption{New partial key generation algorithm}
	\label{New partial key generation algorithm}
\end{algorithm}

In Algorithm \ref{Random matrix generation algorithm}, the matrix elements consist of a set of $B$ bits, requiring the initial encrypted binary string's length to be $B \cdot r \cdot r$. Each element comprises $B$ bits from least to most significant, multiplied by $q/2^B$ to determine its integer value. These random matrix elements are filled row-by-row from the least to the most significant entry. For step 3, to ensure all vectors in the final key matrix are linearly independent, it's essential to filter the generated matrix, eliminating any linearly correlated vectors.

Algorithm \ref{New partial key generation algorithm} showcases that if a client receives two or more vectors, they can obtain a new vector via linear combination of mutually perpendicular vectors received, remaining orthogonal to other clients' vectors. This can provide clients with improved convenience to update their embedded keys, expand their key space, and considerably heighten the complexity for attackers seeking to obtain these keys. While the actual matrix size used is significantly larger, a smaller dimensional example is presented below for illustrative purposes.

\textit{Example 1:} Consider $client_1$ and $client_2$ selecting $n_1 = 1$ and $n_2 = 2$ respectively, resulting in a 3 $\times$ 3 key matrix. Assuming the server generates the random number 136,777 with $B=2$ and $q=32768$, the initial encrypted binary string is "100001011001001001". Following Algorithm \ref{Random matrix generation algorithm}, the matrix becomes 

$\mathsf{KeyMat}$ = $\frac{32768}{4} \times$ 
\begin{math}
	\begin{bmatrix}
		2 & 1 & 0 \\
		0 & 2 & 2 \\
		1 & 1 & 1
	\end{bmatrix}
\end{math}. 

Post Schmidt Orthogonalization, $\mathsf{KeyMat}$ transforms to

$\mathsf{KeyU}$ = $\frac{32768}{4} \times$ 
\begin{math}
	\begin{bmatrix}
		2 & -1/5 & -4/21 \\
		0 & 2 & -2/21 \\
		1 & 2/5 & 8/21
	\end{bmatrix}
\end{math}.

The server sends $\mathbf{v}_1 = [2,0,1]$ to $client_1$ as $\mathbf{u}_1$, $\mathbf{v}_2 = [-1/5, 2, 2/5]$ and $\mathbf{v}_3 = [-4/21, -2/21, 8/21]$ to $client_2$. Using Algorithm \ref{New partial key generation algorithm}, $client_2$ generates $\mathbf{u}_2 = [-5, 8, 10]$ with selected $coef_1 = 5, coef_2 = 21$. Fig.~\ref{Key Generation of two clients} showcases the key generation example.

\begin{figure}
	\centering
	\includegraphics[]{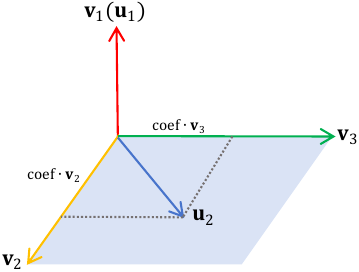}
	\caption{Example of key generation.}
	\label{Key Generation of two clients}
\end{figure}

%
%
%
%
%
%

\section{Theoretical Evaluation}
\label{Performance of Multiparty Reversible Watermarking}
This section delves into a comprehensive theoretical analysis of the performance and robustness of FedReverse. We discuss various metrics used to evaluate the distortion, robustness against interference, perfect covering attributes, and resistance against known original attacks. 

\subsection{Distortion of Embedding}
\label{Analysis of Distortion}
The Mean Square Error (MSE) for a single client is denoted as:
\begin{equation}
	\mathrm{MSE}_i = \frac{\alpha_i^2 2^{2b}\Delta_i^2}{12},
\end{equation}
where each client embeds $b$ bits of information according to \eqref{MSE_DC}. As the projected directions are mutually orthogonal, the embedding MSE of FedReverse between original signals and watermarked signals is represented by:
\begin{equation}
	\mathrm{MSE}_{\mathrm{Fed}} = \sum_{i=1}^{n} \mathrm{MSE}_i  = \frac{2^{2b}}{12}\sum_{i=1}^{n}{\alpha_i^2}{\Delta_i}^2. \label{proposed_MSE}
\end{equation}
Equation \eqref{proposed_MSE} evidently demonstrates that distortion increases with $\alpha$ and $\Delta$, implying a positive correlation between the number of clients and distortion. Therefore, achieving tolerable distortion necessitates an appropriate choice regarding the number of clients.

Another metric for evaluating distortion is the signal-to-watermark ratio (SWR), defined as:
\begin{equation}
	\mathrm{SWR} = 10 \times \log_{10}\left(\frac{\sigma_\mathbf{s}^2}{\sigma_\mathbf{w}^2}\right),
\end{equation}
where $\sigma_\mathbf{s}^2$ and $\sigma_\mathbf{w}^2$ represent the variances of the host signal and the additive watermark, respectively.

\begin{thm}
	The SWR of the proposed watermarking scheme is given by:
	\begin{equation}
		\mathrm{SWR_{Fed}} = -20 \log_{10}\left(\sum_{i=1}^n \left[\alpha_i\beta_i \frac{(\sum_{j=1}^r(\mathbf{u}_i)_j)^2}{r^2\|\mathbf{u}_i\|^2} \right]\right),
	\end{equation}
	where $\beta_i$ represents the reduction ratio for the difference between signals before and after embedding projection, and $j$ is the index of the $j$-th component of $\mathbf{u}_i$.
\end{thm}

\begin{proof}
	Recall that 	$\mathbf{y}  = \mathrm{Emb}_{\mathrm{Fed}}(\mathbf{s}, \mathbf{m}_1, \ldots, \mathbf{m}_n, K_1, \ldots, K_n)$. Let  $y_j$ be the components of $\mathbf{y}$ and $s_j$ be the components of $\mathbf{s}$ separately. Due to Eqs. \eqref{DC-embed} and \eqref{multiple_embed}, we have
	\begin{align}
		y_j &= s_j + \sum_{i=1}^n\Big[ \alpha_i \big[Q_{\mathbf{m},K} \big(\mathrm{proj}_{\mathbf{u}_i}(\mathbf{s}) \big)\big]_j - \alpha_i \big[\mathrm{proj}_{\mathbf{u}_i}(\mathbf{s})\big]_j \Big] \nonumber \\
		&= s_j + \sum_{i=1}^n \Big[ \alpha_i \beta_i \big[\mathrm{proj}_{\mathbf{u}_i}(\mathbf{s})\big]_j \Big],
		\label{component_y}
	\end{align}
	in which $\beta_i \in [0, \Delta_i/2)$ is reduction ratio for the difference between before and after embedding projecting signals.
	Since
	\begin{align}
		\mathrm{proj}_{\mathbf{u}_i}(\mathbf{s}) = \frac{s_1 u_1+\cdots+s_j u_j+\cdots+s_r u_r}{\|\mathbf{u}_i\|^2},
	\end{align}
	\eqref{component_y} can be written as 
	\begin{align}
		\label{linearly_related}
		y_j &= s_j + \sum_{i=1}^n \Big[\alpha_i\beta_i \frac{(\mathbf{u}_i)_j^2}{\|\mathbf{u}_i\|^2} \Big]\cdot s_j \nonumber \\
		&+ \sum_{i=1}^n \Big[\alpha_i\beta_i \frac{s_1(\mathbf{u}_i)_1 + \cdots + s_r(\mathbf{u}_i)_r }{\|\mathbf{u}_i\|^2}(\mathbf{u}_i)_j \Big].
	\end{align}
Let the variance of the host signal $\mathbf{s}$ be
	\begin{align}
		\label{variance_s}
		\sigma_\mathbf{s}^2  
		&= \frac{1}{r}\sum_{j=1}^r\mathrm{Var}(s_j).
	\end{align}
With reference to \eqref{linearly_related}, the variance of the additive watermark is
\begin{align}
	\sigma_\mathbf{w}^2  
	&= \Big(\sum_{i=1}^n \Big[\alpha_i\beta_i \frac{(\sum_{j=1}^r(\mathbf{u}_i)_j)^2}{r^2\|\mathbf{u}_i\|^2} \Big] \Big)^2 	\sigma_\mathbf{s}^2.
\end{align}
Therefore the SWR of FedReverse can be calculated as
\begin{equation}
	\mathrm{SWR_{Fed}} = -20 \log_{10}{\Big(\sum_{i=1}^n \Big[\alpha_i\beta_i \frac{(\sum_{j=1}^r(\mathbf{u}_i)_j)^2}{r^2\|\mathbf{u}_i\|^2} \Big]\Big)}.
\end{equation}
\end{proof}

\subsection{Robustness Against Interference}
Watermarking should exhibit resilience against watermark removal attacks aimed at creating surrogate models capable of bypassing provenance verification. Numerous watermarking schemes asserting robustness have been introduced, including adversarial training \cite{DBLP:conf/iclr/MadryMSTV18}, feature permutation \cite{DBLP:conf/sp/LukasJLK22}, fine-pruning \cite{DBLP:conf/raid/0017DG18}, fine-tuning \cite{uchida2017embedding}, weight pruning \cite{DBLP:conf/iclr/ZhuG18}, regularization \cite{DBLP:conf/ih/ShafieinejadLWL21}, etc. 

We consider the additive interference model as $\tilde{\mathbf{y}}= \mathbf{y} + \mathbf{n}$, with $\mathbf{n}$ being the interference. The additive noise model serves as a conceptual framework to encompass these diverse attacks, portraying them as perturbations $\mathbf{n}$ added to the watermarked signal $\mathbf{y}$.
The consideration of the additive noise model helps abstract various attacks that attempt to undermine watermarking by introducing interference or perturbations to the watermarked signal. 
Hereby we show that FedReverse enjoys the robustness of watermarked messages. 

\begin{thm}
	The watermarked messages can be faithfully recovered if the additive attacks $\mathbf{n}$ satisfies
	\begin{align}
		|\mathrm{proj}_{\mathbf{u}_i}(\mathbf{n})| \in \Big[0, \frac{(2\alpha_i - 1)\Delta_i}{4}\Big].
	\end{align}
\end{thm}

\begin{proof}
	
	According to \eqref{multiple_extract}, the received interfered signal $\tilde{\mathbf{y}}$ is projected onto the $i$-th client's vector $\mathbf{u}_i$ before extracting, i.e.,
	\begin{equation}
		\mathrm{proj}_{\mathbf{u}_i}(\tilde{\mathbf{y}}) = \mathrm{proj}_{\mathbf{u}_i}(\mathbf{y}) + \mathrm{proj}_{\mathbf{u}_i}(\mathbf{n}).
	\end{equation}
	Due to the periodicity and symmetry of embedding function, $|\mathrm{proj}_{\mathbf{u}_i}(\mathbf{s})| \in [0, \Delta_i/2]$. The vector form of \eqref{linearly_related} can be abbreviated as
	\begin{align}
		\mathrm{proj}_{\mathbf{u}_i}(\mathbf{y}) = c\cdot \mathrm{proj}_{\mathbf{u}_i}(\mathbf{s}).
	\end{align}
	Further, we can obtain $|\mathrm{proj}_{\mathbf{u}_i}(\mathbf{y})| \in [0, (1-\alpha_i)\Delta_i/2]$.
	
	As mentioned in Section~\ref{Difference Contraction}, extracting watermark is to find the closest coset $\mathbf{\Lambda}_{\mathbf{m}_i}$ to $\mathrm{proj}_{\mathbf{u}_i}(\tilde{\mathbf{y}})$. Specially, correct extracting requires that $\mathrm{proj}_{\mathbf{u}_i}(\tilde{\mathbf{y}})$ is in $\mathbf{\Lambda}_{\mathbf{m}_i}$ where $\mathrm{proj}_{\mathbf{u}_i}(\mathbf{y})$ stays, namely $|\mathrm{proj}_{\mathbf{u}_i}(\tilde{\mathbf{y}})| \in [0, \Delta_i/4]$. A intuitive explanation is shown in Fig.\ref{one_party}, and the projecting interference $|\mathrm{proj}_{\mathbf{u}_i}(\mathbf{n})| = |\mathrm{proj}_{\mathbf{u}_i}(\tilde{\mathbf{y}})| - |\mathrm{proj}_{\mathbf{u}_i}(\mathbf{y})| \in [0, (2\alpha_i-1)\Delta_i/4]$. Correspondingly, there is no demand to consider the interference in the direction without any client.
\end{proof}

\begin{figure}[t]
	\centering
	\includegraphics[width=0.5\textwidth]{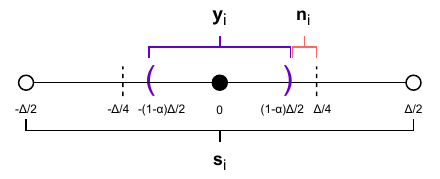}
	\caption{The relationship among the projection $\mathbf{s}_i, \mathbf{y}_i$ and $\mathbf{n}_i$. The correct extracting range is $(-\Delta/4, \Delta/4)$.}
	\label{one_party}
\end{figure}

\subsection{Perfect Covering}
Confidentiality stands as a foundational principle in ensuring security and privacy within machine learning, as emphasized by Papernot et al. \cite{DBLP:conf/eurosp/PapernotMSW18}. The concept of \textit{perfect secrecy}, originally pioneered by Shannon \cite{shannon1949communication}, has been instrumental in defining the notion of  perfect covering in watermarking \cite{cayre2005watermarking}.  Mathematically, this condition is expressed as:
\begin{align}
	p_\mathbf{W}(\mathbf{w}) = p_\mathbf{W}(\mathbf{w}|\mathbf{y}),\ \text{for any} \ (\mathbf{y}, \mathbf{w}),
\end{align}
where $p_\mathbf{W}(\cdot)$ represents the probability density function of the watermark. This condition ensures that the presence of the watermark remains completely concealed within the watermarked content, even when observing the content itself, thereby preserving the secrecy and integrity of the embedded information.


\begin{thm}
	The proposed FedReverse provides perfect covering.
\end{thm}

 \begin{proof}
 	Considering the projection on one client $\mathbf{u}_i$, due to the principle of QIM method, $\mathbf{s}_i$ is quantized into the corresponding coset $\Lambda_{\mathbf{m}_i}$ based on $\mathbf{m}_i$. In other words, the watermarked signal $\mathbf{y}_i$ in $\Lambda_{\mathbf{m}_i}$ represents the hidden message corresponding to $\Lambda_{\mathbf{m}_i}$. However, the hidden messages corresponding to cosets are selected by clients, which is unknown to attackers. Also, $\Delta$ is related with quantization, which is selected by clients and unaware for attackers as well. Hence leaking information about watermark does not increase the likelihood of accessing to $\mathbf{y}_i$. Further, global watermarked signal $\mathbf{y}$ is secure.

 	Assume that the probability density function of $\mathbf{y}$ is $p_\mathbf{Y}(\mathbf{y})$. Based on the above inference, we have
 	\begin{align}
 		p_\mathbf{Y}(\mathbf{y}|\mathbf{w}) = p_\mathbf{Y}(\mathbf{y}),
 	\end{align}
 	where additive watermark $\mathbf{w}$ is calculated by all of $\mathbf{m}_i$ and $\mathbf{s}_i$.
 	Based on Bayes rule $p_\mathbf{Y}(\mathbf{y}|\mathbf{w})p_\mathbf{W}(\mathbf{w}) = p_\mathbf{W}(\mathbf{w}|\mathbf{y})p_\mathbf{Y}(\mathbf{y})$, then $p_\mathbf{W}(\mathbf{w}) = p_\mathbf{W}(\mathbf{w}|\mathbf{y})$.
 \end{proof}

\subsection{Resistance to Known Original Attack}
We define the \textit{Known Original Attack} (KOA) as the security level of the watermarked scheme, following Diffie-Hellman's Terminology \cite{katz2007introduction}. KOA represents a scenario in which an attacker obtains $N_o$ watermarked vectors along with their corresponding original versions, forming pairs $(\mathbf{y}, \mathbf{s})^{N_o}$.

The primary objective for the attacker in this scenario is to extract the watermark by inferring the key rather than restoring the model. The attacker can calculate the difference vector $\mathbf{e} = \mathbf{y} - \mathbf{s}$ based on the obtained information. According to \eqref{multiple_embed}, the attacker tends to assume that a longer vector is more affected by embedding in all directions. By decomposing the longest vector $\mathbf{e}$ into $r$ pairwise vertical vectors as the client vectors $\mathbf{u}$ and adjusting the direction according to the remaining $\mathbf{e}$, the attacker infers the key $K'$. Assuming that the inferred key $K'$ is the same as or close to the actual key $K$ due to the unknown cosets, the conditional entropy is given by the formula:
\begin{align}
    H(K' | (\mathbf{y}, \mathbf{s})^{N_o}) &= - \sum_{K'} P(K' | (\mathbf{y}, \mathbf{s})^{N_o}) \cdot \log  P(K' | (\mathbf{y}, \mathbf{s})^{N_o}) \\
    &= \frac{n}{|\mathcal{M}|!} \log_{|\mathcal{M}|}\big(\frac{1}{|\mathcal{M}|!}\big),
\end{align}
where \( P(K' | (\mathbf{y}, \mathbf{s})^{N_o}) \) represents the conditional probability distribution of the inferred key \( K' \) given the observed pairs \( (\mathbf{y}, \mathbf{s})^{N_o} \), and $\mathcal{M}$ denotes the message space of $\mathbf{m}_1, \ldots, \mathbf{m}_n$. This formula calculates the average uncertainty or information content about the inferred key \( K' \) after observing the set of pairs. The entropy is a measure of the average amount of  uncertainty associated with the inferred key \( K' \) given the observed data.
Nevertheless, it is impossible to correctly decompose $\mathbf{e}$ to obtain $K$ in $\mathbb{R}^r$ though the robustness allows a certain amount of offset. Therefore
\begin{align}
    H(K | (\mathbf{y}, \mathbf{s})^{N_o}) \gg \frac{n}{|\mathcal{M}|!} \log_{|\mathcal{M}|}\big(\frac{1}{|\mathcal{M}|!}\big).
\end{align}

Moreover, based on KOA, we examine the existential unforgeability of our proposed watermarking scheme as follow:

\textbf{Existential Unforgeability under the Known Orignal Attack (EUF-KOA)}: Similarly to Existential Unforgeability under Chosen-Plaintext Attack(EUF-CPA) \cite{katz2007introduction}, the security of which can be considered as the advantage of the attackers forging signatures for plaintext when they obtain the public key, the security of EUF-KOA in our proposed watermarking scheme is the possibility of the attackers forging one certain client for watermark embedding. If the advantage of the attackers can be ignored for any polynomial time, the scheme is considered as EUF-KOA secure.

Thus, the EUF-KOA security in the proposed watermarking scheme can be defined as a game between attacker $\mathcal{A}$ and challenger $\mathcal{C}$ as follow:

\noindent\textbf{Setup}: $\mathcal{A}$ obtains the cover vector $s$.
    
\noindent\textbf{Query Phase}: In this stage, $\mathcal{A}$ can conduct a series of queries to obtain a series of $\mathbf{y}$: when $\mathcal{A}$ submits one $\mathbf{m}_i$, $\mathcal{C}$ applies $\mathrm{Emb}_{\mathrm{DC}}(\mathbf{s}, \mathbf{m}_i, K_i)$ and gives the generated $\mathbf{y}_i$ to $\mathcal{A}$.
    
\noindent\textbf{Output}: $\mathcal{A}$ outputs $(\mathbf{m}_i^{*},\mathbf{y}_i^{*})$.
    
If $\mathrm{Emb}_{\mathrm{DC}}(\mathbf{s}, \mathbf{m}_i^{*},K_i) = \mathbf{y}_i^{*}$ and $\mathcal{A}$ never queries $\mathbf{m}_i^{*}$ in \textit{Query Phase}, $\mathcal{A}$ wins the game. Therefore, we define the attacker's advantage of the EUF-KOA is the probability of attacker $\mathcal{A}$ winning the game.
    
 \begin{thm}
 	The proposed watermarking scheme provides EUF-KOA.
 \end{thm}

 \begin{proof}
 	If $\mathcal{A}$ wants to forge $i$-th client to embed watermark $\mathbf{m}_i^{*}$, $\mathcal{A}$ needs to get $K_i$. However, because of the orthogonality of $\mathbf{u}_i$, only finding correct $\mathbf{u}'_i$ that are in the same direction as $\mathbf{u}_{i}$, $\mathcal{A}$ can forge completely. In other words, $\mathbf{u}'_i = \gamma \mathbf{u}_i$. 
 	
 	Because of the key generation, $\mathbf{u}_i \in \mathbb{R}^r$, the probability to find the satisfied $\mathbf{u}'_i$ is approaching 0. Thus, the probability to find $\mathbf{u}_i$ is approaching 0 as well. Also, the dithers $d_i$  are secret for clients, which usually the attackers cannot obtain.
 	Consequently, without $K_i$, $P(\mathbf{y}_i = \mathbf{y}_{i}^{*})$ = 0, which is to say that the attacker's advantage of the EUF-KOA is 0.
 \end{proof}

\section{Simulations}\label{Sec. experiment}

\subsection{Simulation Settings}

\noindent \textbf{Models and Datasets:} In our experiments, we employ Multi-layer Perceptron (MLP) and Convolutional Neural Networks (CNN) as training models for the image classification task. MLP offers strong expressive and generalization abilities. We construct a 2-layer MLP model with 99,328 weights and CNN models with different layers. Additionally, we utilize the MNIST dataset \cite{lecun1998gradient}, comprising 60,000 training and 10,000 testing grayscale images of handwritten digits. The initial learning rate is set to 0.05, and training lasted for 10 epochs.

\noindent \textbf{Scheme Setups:} We divide the value space into two cosets for each clients, i.e. $\Lambda_{\mathbf{m}_i} \in \{0, 1\}$. For this reason, we can obtain each $\alpha_i \in [0.5, 1)$. The embedding location is chosen as the first layer of each model. Meanwhile, the number of watermarks for each client is modified depending on the number of weights of the first layer and the selected dimension.

\subsection{Test Accuracy}
The comparison between the FedReverse scheme and other watermarking techniques in federated learning, namely FedTracker \cite{shao2022fedtracker} and WAFFLE \cite{tekgul2021waffle}, is illustrated in Figure \ref{compare}. Due to the reversibility of FedReverse, it attains a test accuracy on par with the original model, surpassing the performance of both WAFFLE and FedTracker. For instance, when the number of clients $n=20$, the respective test accuracies for FedTracker, WAFFLE, FedReverse, and the scenario without watermarking are recorded at $0.9833$, $0.9829$, $0.9918$, and $0.9920$.

Within the trained MLP model, the process of embedding watermarks induces a dispersion of weights within a specified range. Following the removal of these embedded watermarks facilitated by the FedReverse approach, we can obtain the restoration of the original, unaltered weights. This critical observation is visually depicted and substantiated in Figure \ref{weights}.

\begin{figure}[t]
	\centering
	\includegraphics[width=0.45\textwidth]{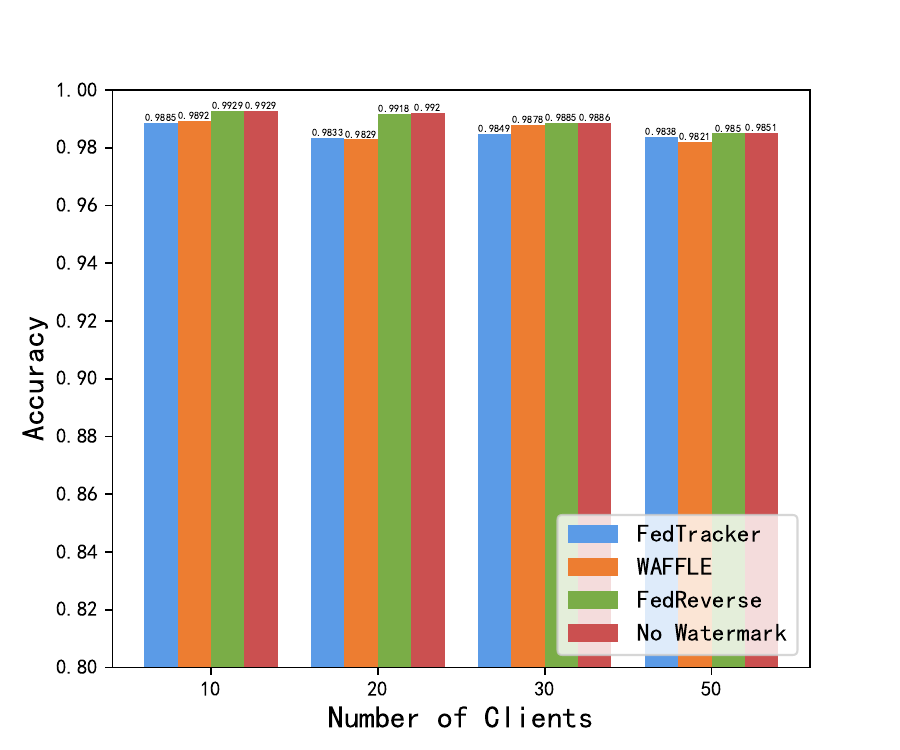}
	\caption{Comparison with other watermarked federated learning schemes.}
	\label{compare}
\end{figure}
\begin{figure}[t]
	\centering
	\includegraphics[width=0.45\textwidth]{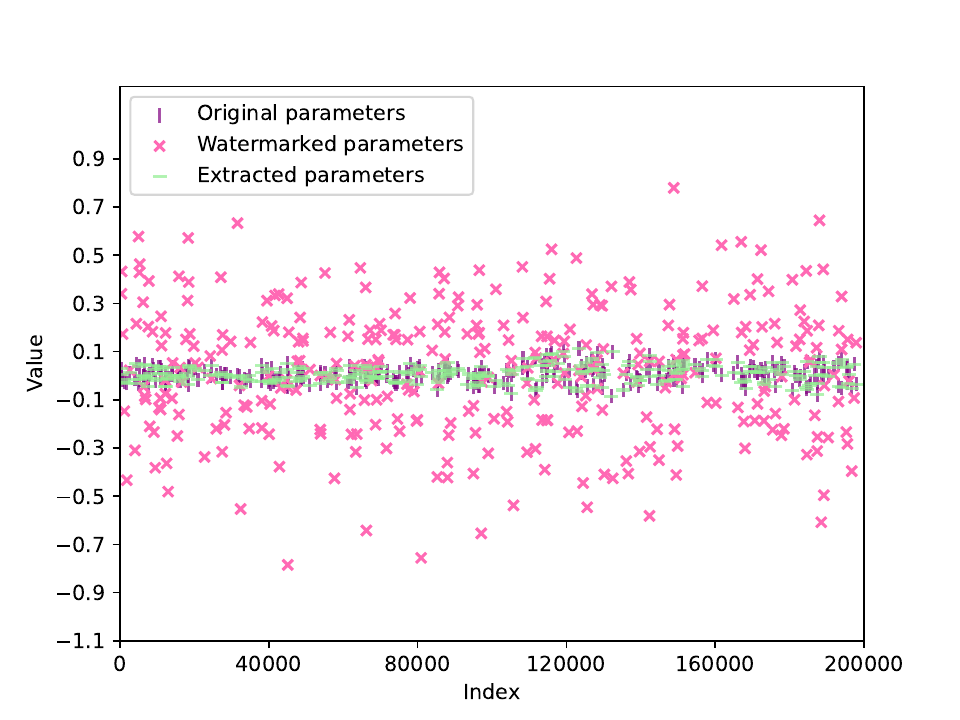}
	\caption{Original, watermarked, and recovered weights.}
	\label{weights}
\end{figure}
 
\subsection{Impact of the Number of Clients}
The influence of the number of clients $n$ plays a pivotal role in the context of our investigation. To scrutinize this aspect, we conducted tests to gauge the impact of embedding multiple watermarks into both MLP and CNN weights, with varying numbers of clients participating in the process. For our experiments, we have uniformly set each client's scaling factor to $\Delta = 0.1$.

Accommodating a larger number of clients necessitates the utilization of higher-dimensional signal vectors, inevitably resulting in a trade-off with model accuracy. The findings outlined in Table \ref{number of users} illustrate the discernible variations in average model accuracy concerning the increasing number of clients. Notably, when $n$ rises from $2$ to $8$ in the CNN model, or from $2$ to $10$ in the MLP model, the test accuracies remain around $0.99$ and $0.94$, respectively. It shows that the impact of increasing the number of clients remains relatively limited. Furthermore, owing to the inherent randomness associated with the embedded message, the accuracy portrays certain degrees of fluctuation.

\begin{table}
\caption{Accuracy of trained models with different $n$ under the same $\Delta$ and $\alpha$.}
\label{number of users}
\scalebox{1.34}{
\begin{tabular}{c|c|c|c|c|c}
\hline
Model                & $n$           & $r$           & $\Delta$                    & $\alpha$   & Accuracy \\ \hline
\multirow{6}{*}{CNN} & \multirow{2}{*}{2}  & \multirow{6}{*}{16} & \multirow{6}{*}{all 0.1} & all 0.9 & 0.99156  \\
                     &                     &                     &                          & all 0.5 & 0.99246  \\ \cline{2-2} \cline{5-6} 
                     & \multirow{2}{*}{4}  &                     &                          & all 0.9 & 0.99168  \\
                     &                     &                     &                          & all 0.5 & 0.99070  \\ \cline{2-2} \cline{5-6} 
                     & \multirow{2}{*}{8}  &                     &                          & all 0.9 & 0.98622  \\
                     &                     &                     &                          & all 0.5 & 0.9908   \\ \hline
\multirow{6}{*}{MLP} & \multirow{2}{*}{2}  & \multirow{6}{*}{10} & \multirow{6}{*}{all 0.1} & all 0.9 & 0.9464   \\
                     &                     &                     &                          & all 0.5 & 0.9469   \\ \cline{2-2} \cline{5-6} 
                     & \multirow{2}{*}{5}  &                     &                          & all 0.9 & 0.9443   \\
                     &                     &                     &                          & all 0.5 & 0.9469   \\ \cline{2-2} \cline{5-6} 
                     & \multirow{2}{*}{10} &                     &                          & all 0.9 & 0.9426   \\
                     &                     &                     &                          & all 0.5 & 0.9463   \\ \hline
\end{tabular}
}
\end{table}

\subsection{Impacts of $\alpha$, $\Delta$, $r$}
For the sake of simplicity, we set the default embedding dimension to $r=16$ for CNN and $r=10$ for MLP, while the number of clients is $n=2$. The original CNN exhibits an accuracy of 0.9931, and the original MLP achieves an accuracy of 0.9466. To establish a reference point for assessing the scheme's impact on model performance, we measure the baseline accuracy of both models after embedding the watermark.

Figures \ref{Accuracy_alpha_a} and \ref{Accuracy_delta_b} illustrate the accuracy of watermarked models concerning varying values of $\alpha$ and $\Delta$. Here, CNN2 denotes a 2-layer CNN, CNN4 a 4-layer CNN, and CNN a 6-layer CNN. It is evident from the observations that embedding messages results in a decrease in model accuracy, with higher values of $\alpha$ and $\Delta$ causing a more pronounced decline in accuracy. Notably, the MLP model displays a higher sensitivity to watermark embedding in contrast to CNN. Loosely inferred, a higher-quality model exhibits lesser sensitivity. Additionally, the impact of dimensions on accuracy appears relatively small. All results are summarized in Table \ref{acc_dimension}.

\begin{figure}[]
	\centering
	\includegraphics[width=0.45\textwidth]{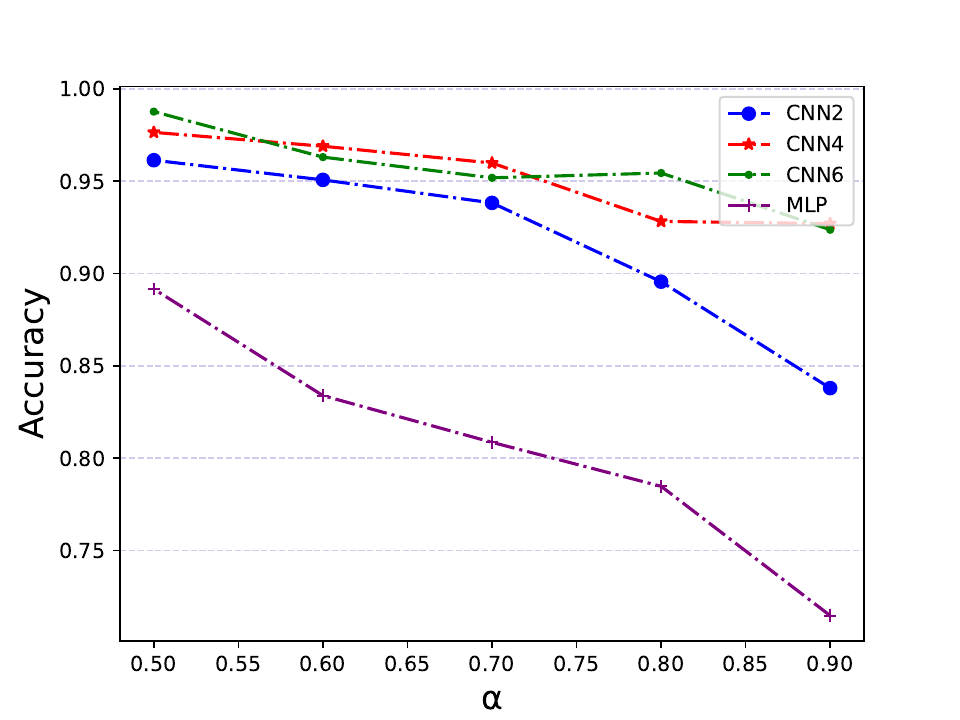}
	\caption{Accuracy of trained models with different $\alpha$, under $\Delta=1$.}
	\label{Accuracy_alpha_a}
\end{figure}

\begin{figure}[]
	\centering
	\includegraphics[width=0.45\textwidth]{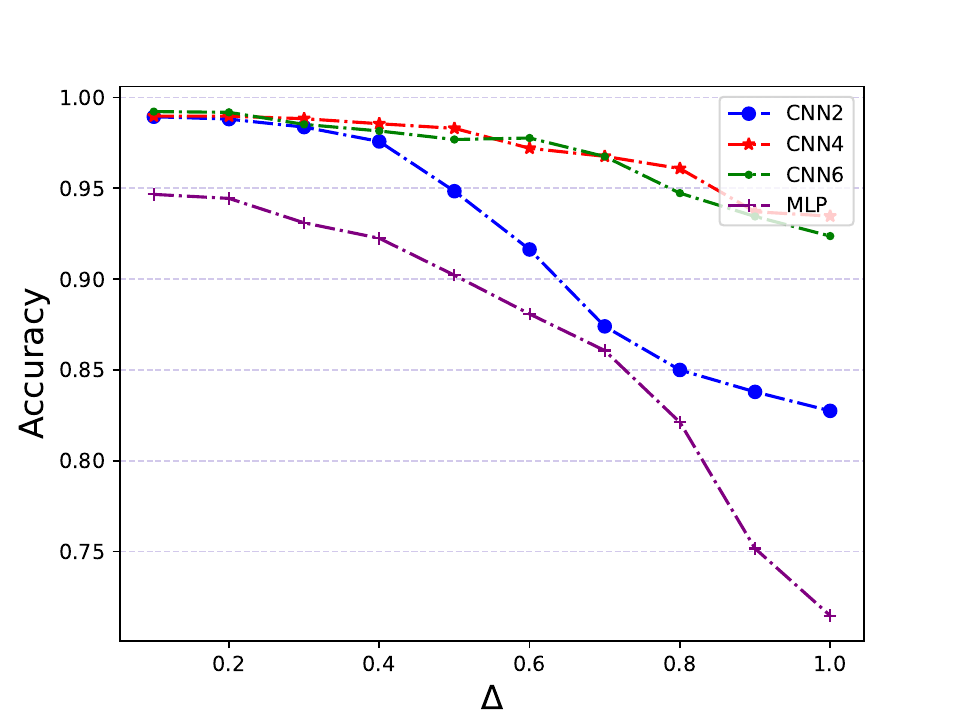}
	\caption{Accuracy of trained models with different $\Delta$ and $\alpha=0.9$.}
	\label{Accuracy_delta_b}
\end{figure}

\begin{table}[]
	\centering
	\caption{Accuracy of trained models with different $r$ under the same $\Delta$ and $\alpha$.}
	\label{acc_dimension}
	\scalebox{1.3}{
		\begin{tabular}{c|c|c|c|c|c}
			\hline
			Model                 & $n$           & $r$           & $\Delta$                & $\alpha$                    & Accuracy \\ \hline
			\multirow{15}{*}{CNN} & \multirow{15}{*}{2} & \multirow{5}{*}{4}  & \multirow{15}{*}{0.1} & 0.9                      & 0.9904   \\
			&                     &                     &                       & 0.8                      & 0.9904   \\
			&                     &                     &                       & 0.7                      & 0.9918   \\
			&                     &                     &                       & 0.6                      & 0.9912   \\
			&                     &                     &                       & 0.5                      & 0.9927   \\ \cline{3-3} \cline{5-6} 
			&                     & \multirow{5}{*}{8}  &                       & 0.9                      & 0.9912   \\
			&                     &                     &                       & 0.8                      & 0.9927   \\
			&                     &                     &                       & 0.7                      & 0.9924   \\
			&                     &                     &                       & 0.6                      & 0.9924   \\
			&                     &                     &                       & 0.5                      & 0.9925   \\ \cline{3-3} \cline{5-6} 
			&                     & \multirow{5}{*}{16} &                       & 0.9                      & 0.9922   \\
			&                     &                     &                       & 0.8                      & 0.9922   \\
			&                     &                     &                       & 0.7                      & 0.9921   \\
			&                     &                     &                       & 0.6                      & 0.9925   \\
			&                     &                     &                       & 0.5                      & 0.9924   \\ \hline
			\multirow{10}{*}{MLP} & \multirow{10}{*}{2} & \multirow{5}{*}{5}  & \multirow{10}{*}{0.1} & \multicolumn{1}{l|}{0.9} & 0.9466   \\
			&                     &                     &                       & \multicolumn{1}{l|}{0.8} & 0.9466   \\
			&                     &                     &                       & \multicolumn{1}{l|}{0.7} & 0.9468   \\
			&                     &                     &                       & \multicolumn{1}{l|}{0.6} & 0.9468   \\
			&                     &                     &                       & \multicolumn{1}{l|}{0.5} & 0.9471   \\ \cline{3-3} \cline{5-6} 
			&                     & \multirow{5}{*}{10} &                       & \multicolumn{1}{l|}{0.9} & 0.9467   \\
			&                     &                     &                       & \multicolumn{1}{l|}{0.8} & 0.9469   \\
			&                     &                     &                       & \multicolumn{1}{l|}{0.7} & 0.9470   \\
			&                     &                     &                       & \multicolumn{1}{l|}{0.6} & 0.9466   \\
			&                     &                     &                       & \multicolumn{1}{l|}{0.5} & 0.9470   \\ \hline
		\end{tabular}
	}
\end{table}

\subsection{Distortion in Watermark Embedding}
The MSE and SWR serve as pivotal metrics for quantifying the efficacy of watermark embedding. In this evaluation, we analyze the MSE and SWR of both MLP and CNN models subsequent to the watermark embedding process, employing distinct values of $\Delta$ and $\alpha$.

For the conducted experiments, we fix the number of clients as $n=1$ for both trained models. Employing an embedding dimension of $r=4$ for CNN and $r=5$ for MLP, Fig.\ref{Alpha_CNN} and Fig.\ref{Alpha_MLP} elucidate the MSE and SWR of watermarked models with $\Delta = 0.1$, showcasing variations in performance concerning different $\alpha$ values. Simultaneously, Fig.\ref{Delta_MSE} and Fig.\ref{Delta_SWR} demonstrate the MSE and SWR of watermarked models with $\alpha = 0.9$ and varying $\Delta$. Notably, the empirical findings reveal a consistent trend: an increase in $\alpha$ and $\Delta$ corresponds to an escalation in MSE while concurrently leading to a decline in SWR. These observed trends align closely with the theoretical analysis expounded in Section~\ref{Analysis of Distortion}.

Comprehensive insights into the impact of $\alpha$ and dimension $r$ on the performance of watermark embedding are tabulated in Table.\ref{Alpha_MSEANDSWR}. Notably, it is discerned that employing a higher dimension augments the efficacy of watermark embedding. Moreover, Table.\ref{num_client_MSEANDSWR} delves into the effect of varying the number of clients ($n$) while maintaining fixed $\alpha$ and $\Delta$. Noteworthy observations manifest that an increase in the number of clients leads to a degradation in the performance of watermark embedding, as evidenced by amplified MSE and diminished SWR metrics.

\begin{figure}[]
	\centering
	\includegraphics[width=.9\linewidth]{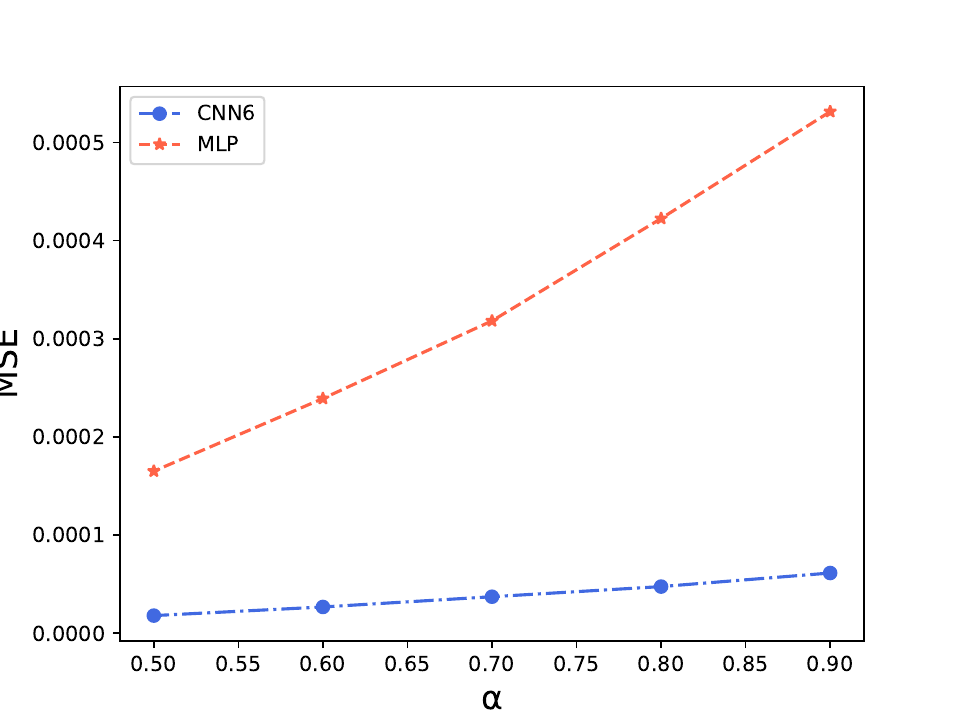}
	\caption{MSE of trained models with different $\alpha$.}
	\label{Alpha_CNN}
\end{figure}

\begin{figure}[]
	\centering
	\includegraphics[width=.9\linewidth]{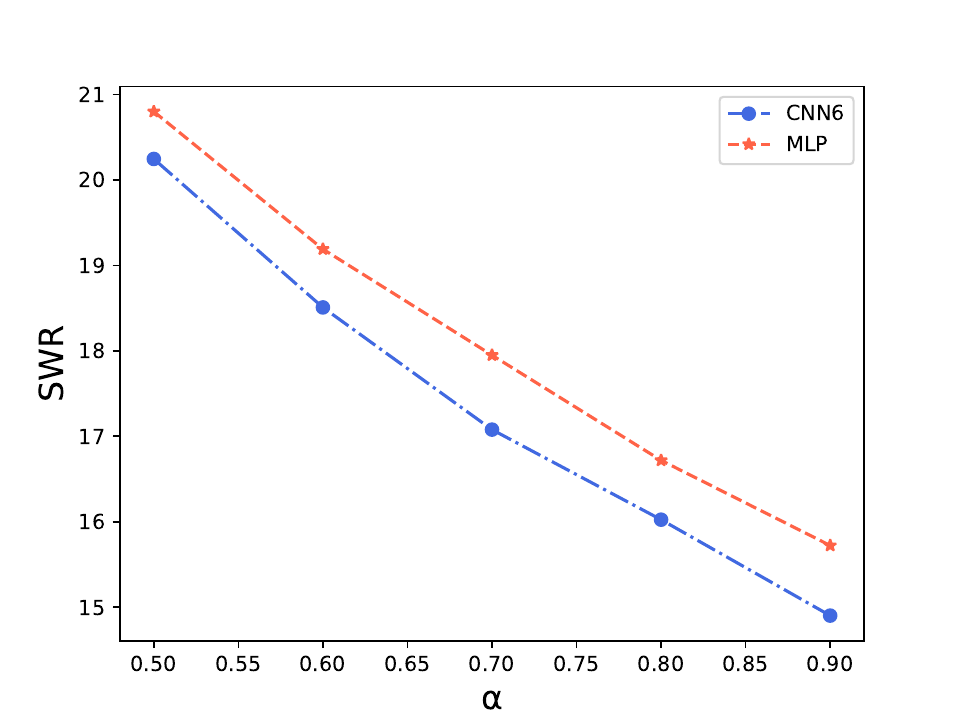}
	\caption{SWR of trained models with different $\alpha$.}
	\label{Alpha_MLP}
\end{figure}

\begin{figure}[]
	\centering
	\includegraphics[width=.85\linewidth]{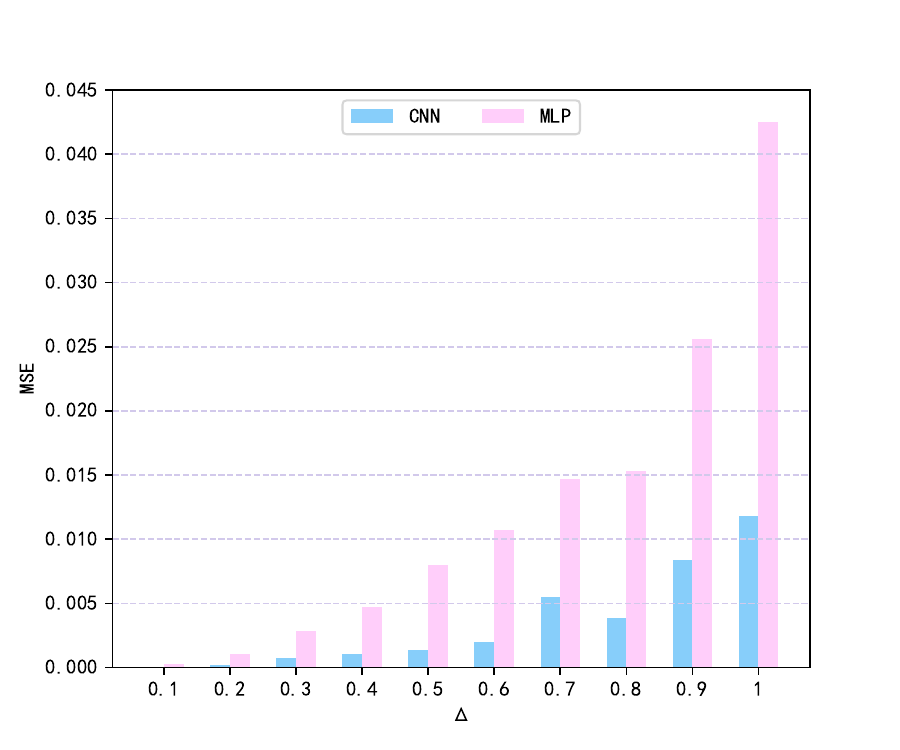}
	\caption{MSE of trained models with different $\Delta$.}
	\label{Delta_MSE}
\end{figure}

\begin{figure}[]
	\centering
	\includegraphics[width=.85\linewidth]{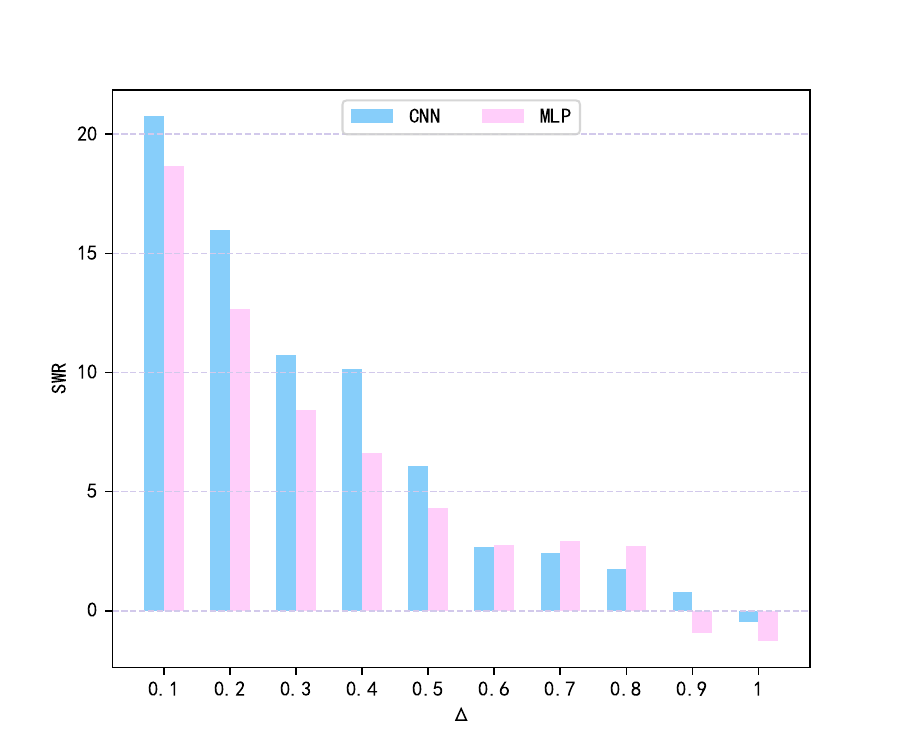}
	\caption{SWR of trained models with different $\Delta$.}
	\label{Delta_SWR}
\end{figure}

\begin{table}[]
    \caption{MSE and SWR of trained models with different $\alpha$ and $r$ with $\Delta$ = 0.1 and $n$ = 1.}
    \label{Alpha_MSEANDSWR}
    \scalebox{1.3}{
        \begin{tabular}{c|c|c|c|c|c}
        \hline
        Model                 & $r$           & $\Delta$                 & $\alpha$ & MSE($10^{-5}$) & SWR \\ \hline
        \multirow{15}{*}{CNN} & \multirow{5}{*}{4}  & \multirow{15}{*}{0.1} & 0.9   &  6.13   &  14.9015   \\
                              &                     &                       & 0.8   & 4.75    &  16.0240   \\
                              &                     &                       & 0.7   &  3.72   &  17.0781   \\
                              &                     &                       & 0.6   &  2.67   &  18.5080   \\
                              &                     &                       & 0.5   &  1.80   &  20.2457   \\ \cline{2-2} \cline{4-6} 
                              & \multirow{5}{*}{8}  &                       & 0.9   &  3.25   &  17.6661   \\
                              &                     &                       & 0.8   & 2.50    &  18.8163   \\
                              &                     &                       & 0.7   &  1.87   &  20.0675   \\
                              &                     &                       & 0.6   & 1.28    & 21.6899    \\
                              &                     &                       & 0.5   &  0.930   &  23.1052   \\ \cline{2-2} \cline{4-6} 
                              & \multirow{5}{*}{16} &                       & 0.9   &  1.60   &  20.7574   \\
                              &                     &                       & 0.8   &  0.14   &   22.2349  \\
                              &                     &                       & 0.7   &  0.942   &  23.0446   \\
                              &                     &                       & 0.6   &  0.648   & 24.6943    \\
                              &                     &                       & 0.5   &  0.413   & 26.6315    \\ \hline
        \multirow{10}{*}{MLP} & \multirow{5}{*}{5}  & \multirow{10}{*}{0.1} & 0.9   & 51.3    & 15.7220    \\
                              &                     &                       & 0.8   &  42.3   & 16.7182    \\
                              &                     &                       & 0.7   &  31.8   & 17.9493    \\
                              &                     &                       & 0.6   &  29.3   & 19.1903    \\
                              &                     &                       & 0.5   &  16.3   & 20.7999    \\ \cline{2-2} \cline{4-6} 
                              & \multirow{5}{*}{10} &                       & 0.9   &  27.1   & 18.6414    \\
                              &                     &                       & 0.8   &  21.1   & 19.7393    \\
                              &                     &                       & 0.7   &  16.1   & 20.9074    \\
                              &                     &                       & 0.6   & 11.7    & 22.2839    \\
                              &                     &                       & 0.5   & 8.18    & 23.8472    \\ \hline
        \end{tabular}
    }
\end{table}

\begin{table}[]
\caption{MSE and SWR of trained models with different $n$ under constant $\alpha$ and $\Delta$.}
\label{num_client_MSEANDSWR}
\scalebox{1.07}{
\begin{tabular}{c|c|c|c|c|c|c}
\hline
Model                & $n$           & $r$          & $\Delta$                    & $\alpha$   & MSE($10^{-5}$)     & SWR     \\ \hline
\multirow{6}{*}{CNN} & \multirow{2}{*}{2}  & \multirow{6}{*}{16} & \multirow{6}{*}{all 0.1} & all 0.9 & 5.11 & 16.6704 \\
                     &                     &                     &                          & all 0.5 & 1.53 & 21.9094 \\ \cline{2-2} \cline{5-7} 
                     & \multirow{2}{*}{4}  &                     &                          & all 0.9 & 9.78 & 13.8584 \\
                     &                     &                     &                          & all 0.5 & 2.92 & 19.1267 \\ \cline{2-2} \cline{5-7} 
                     & \multirow{2}{*}{8}  &                     &                          & all 0.9 & 19.6 & 10.8267 \\
                     &                     &                     &                          & all 0.5 & 6.14 & 15.8744 \\ \hline
\multirow{6}{*}{MLP} & \multirow{2}{*}{2}  & \multirow{6}{*}{10} & \multirow{6}{*}{all 0.1} & all 0.9 & 53.5 & 15.6923 \\
                     &                     &                     &                          & all 0.5 & 16.4 & 20.8117 \\ \cline{2-2} \cline{5-7} 
                     & \multirow{2}{*}{5}  &                     &                          & all 0.9 & 133 & 11.7277 \\
                     &                     &                     &                          & all 0.5 & 41.1 & 16.8370 \\ \cline{2-2} \cline{5-7} 
                     & \multirow{2}{*}{10} &                     &                          & all 0.9 & 267 & 8.7037  \\
                     &                     &                     &                          & all 0.5 & 82.5 & 13.8106 \\ \hline
\end{tabular}
}
\end{table}

\subsection{Weight Distribution Analysis via Histograms}
For a comprehensive understanding of the variations in model weights, we examine the histograms of trained models before embedding, after embedding watermarks, and post-recovery.
Fig.\ref{cnn_Histogram_orignal} and Fig.\ref{mlp_Histogram_orignal} portray the histograms of the original models, which remain consistent with the histograms post-recovery. Furthermore, Table.\ref{Histogram_CNN} and Table.\ref{Histogram_MLP} provide a visual representation of the histograms for the trained models. Notably, the histograms offer several discernible insights: 

\noindent 1) The number of clients and the dimension exhibit negligible impact on model weights, a finding consistent with the earlier conclusion. 
\noindent 2) $\Delta$ and $\alpha$ represent pivotal factors influencing alterations in model weights. As $\Delta$ and $\alpha$ increase, the magnitude of alterations in model weights accentuates, corroborating the aforementioned conclusion.

Hence, for optimal model training tailored to clients' needs, maintaining $\Delta$ and $\alpha$ within rational ranges for each client emerges as a crucial consideration.

\begin{figure}
    \centering
    \includegraphics[width=0.8\linewidth]{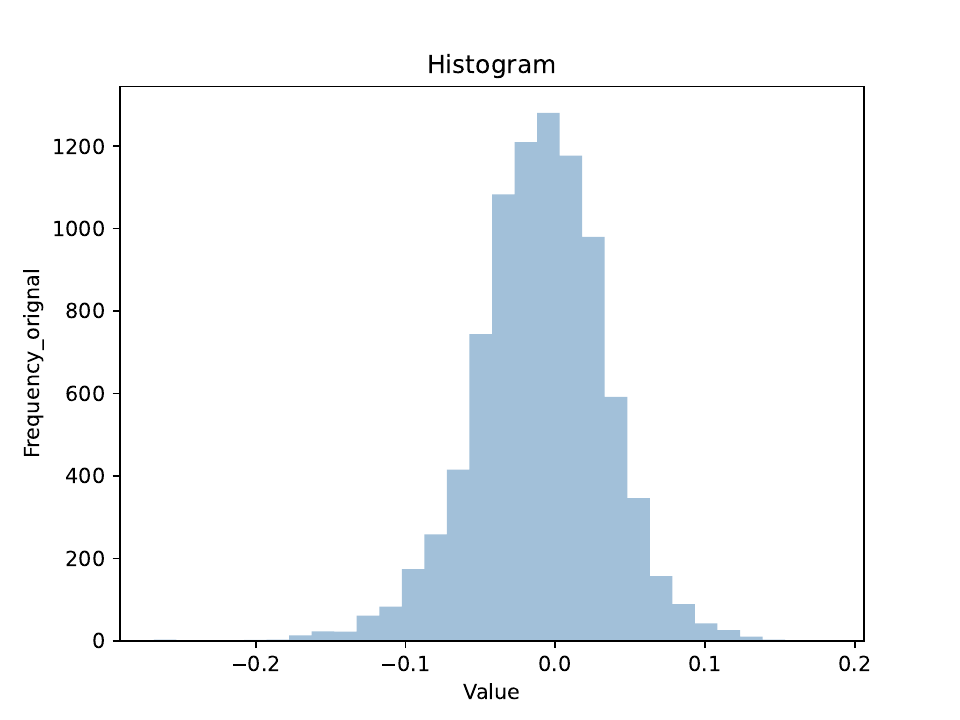}
    \caption{Histogram of the original CNN model.}
    \label{cnn_Histogram_orignal}
\end{figure}

\begin{figure}
    \centering
    \includegraphics[width=0.8\linewidth]{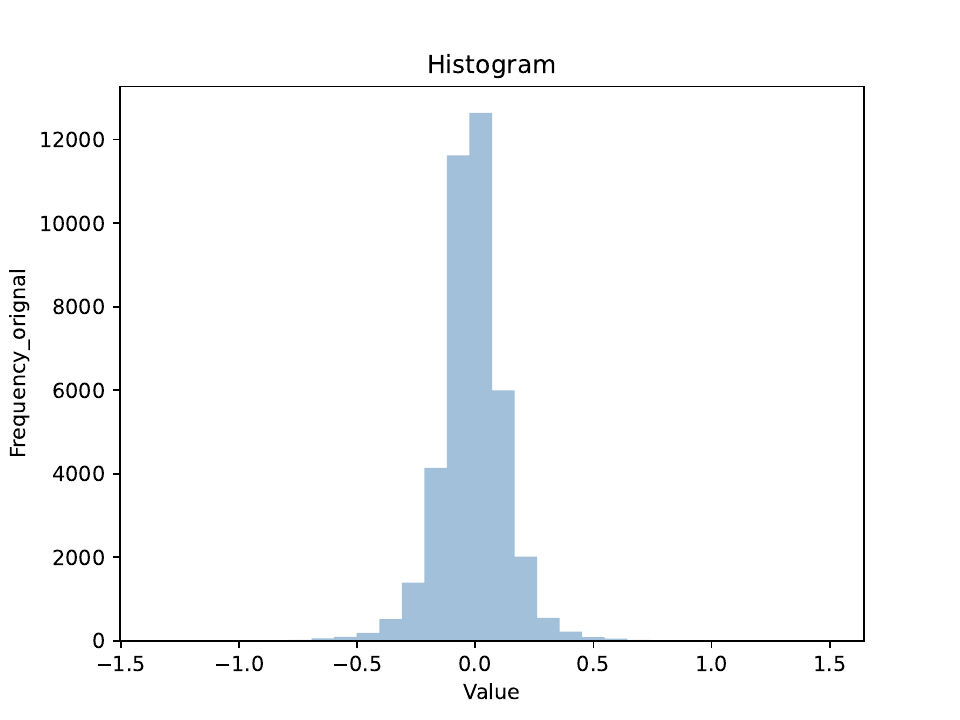}
    \caption{Histogram of the original MLP model.}
    \label{mlp_Histogram_orignal}
\end{figure}

\begin{table*}
\centering
\caption{ Histogram of watermarked CNN in different situation.}
\label{Histogram_CNN}
\begin{tabular}{ccccc}
\hline
\multicolumn{1}{c|}{$n$}               & \multicolumn{4}{c}{1}                                                                                             \\ \hline
\multicolumn{1}{c|}{$r$}               & \multicolumn{4}{c}{4}                                                                                             \\ \hline
\multicolumn{1}{c|}{$\Delta$}           & \multicolumn{2}{c|}{0.1}                                    & \multicolumn{2}{c}{1}                               \\ \hline
\multicolumn{1}{c|}{$\alpha$}           & \multicolumn{1}{c|}{0.9}     & \multicolumn{1}{c|}{0.5}     & \multicolumn{1}{c|}{0.9}     & 0.5                  \\ \hline
\multicolumn{1}{c|}{} & \multicolumn{1}{c|}{\includegraphics[width=.2\linewidth]{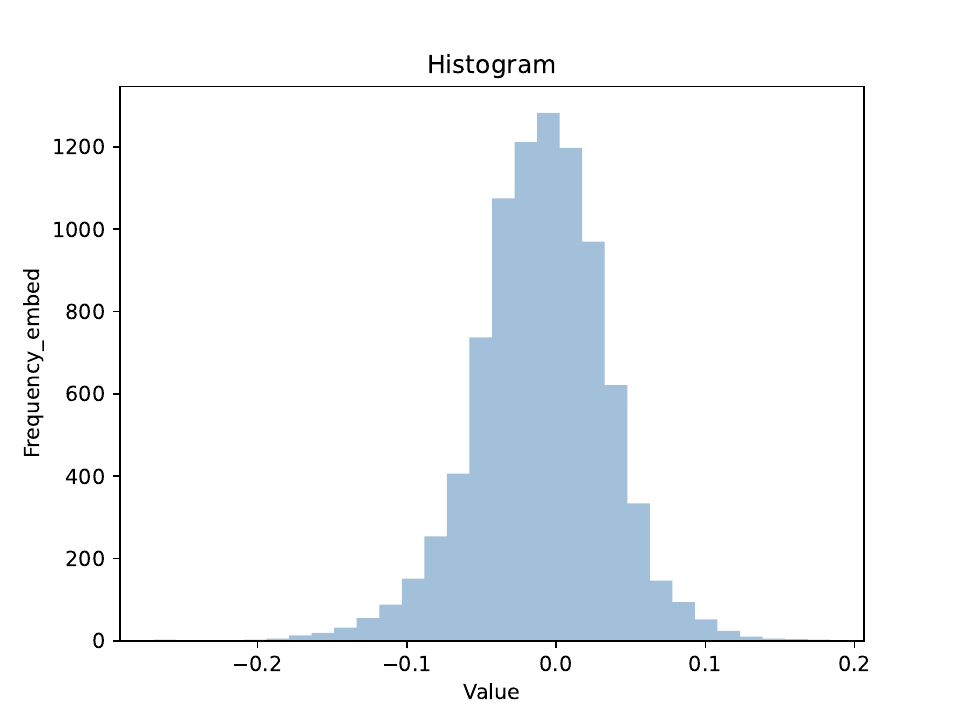} }        & \multicolumn{1}{c|}{\includegraphics[width=.2\linewidth]{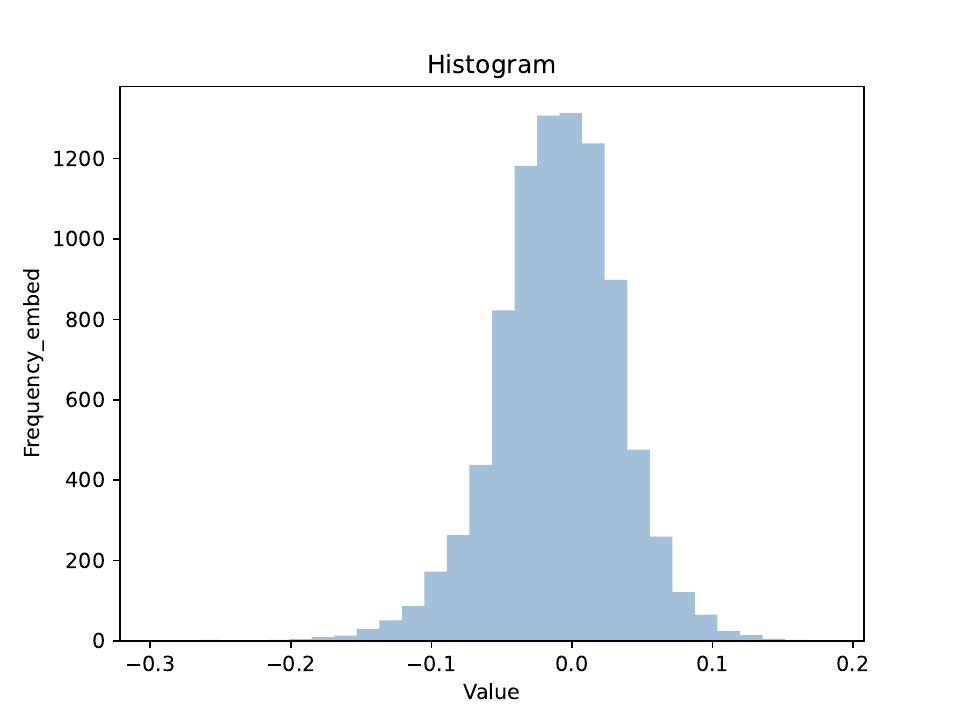} }        & \multicolumn{1}{c|}{\includegraphics[width=.2\linewidth]{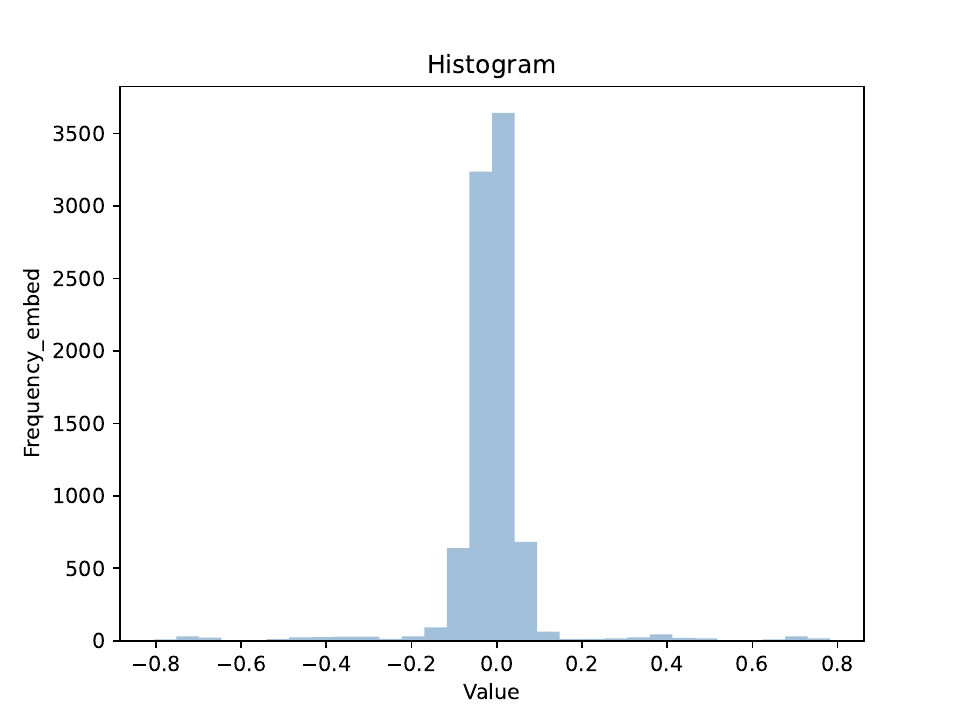} }        &         \includegraphics[width=.2\linewidth]{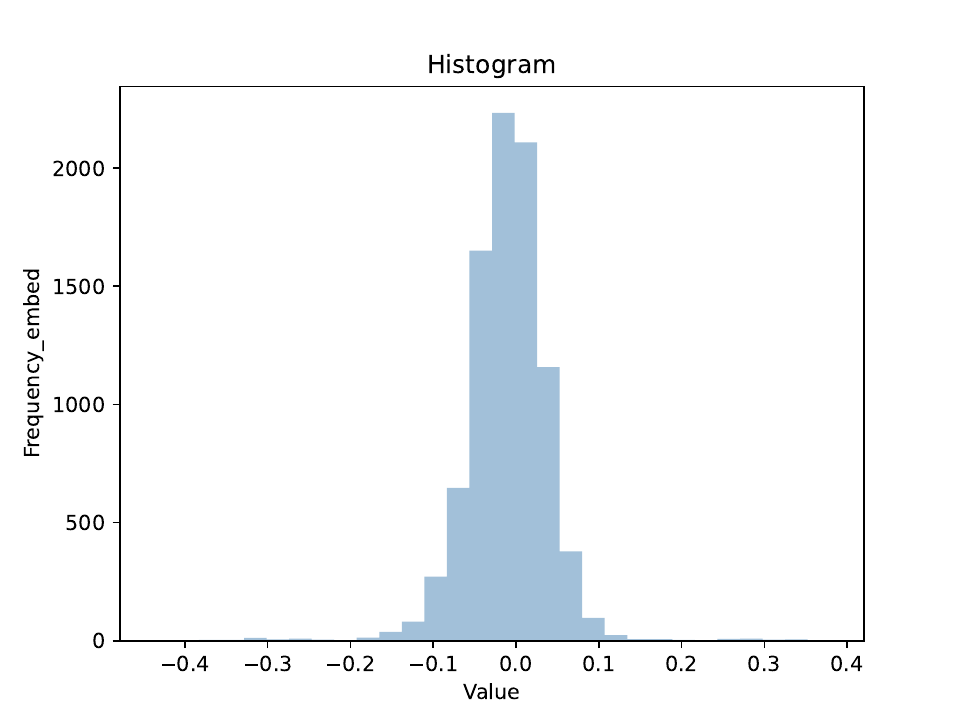}              \\ \hline \hline

\multicolumn{1}{c|}{$n$}               & \multicolumn{4}{c}{4}                                                                                             \\ \hline
\multicolumn{1}{c|}{$r$}               & \multicolumn{4}{c}{4}                                                                                             \\ \hline
\multicolumn{1}{c|}{$\Delta$}           & \multicolumn{2}{c|}{all 0.1}                                & \multicolumn{2}{c}{all 1}                           \\ \hline
\multicolumn{1}{c|}{$\alpha$}           & \multicolumn{1}{c|}{all 0.9} & \multicolumn{1}{c|}{all 0.5} & \multicolumn{1}{c|}{all 0.9} & all 0.5              \\ \hline
\multicolumn{1}{c|}{} & \multicolumn{1}{c|}{\includegraphics[width=.2\linewidth]{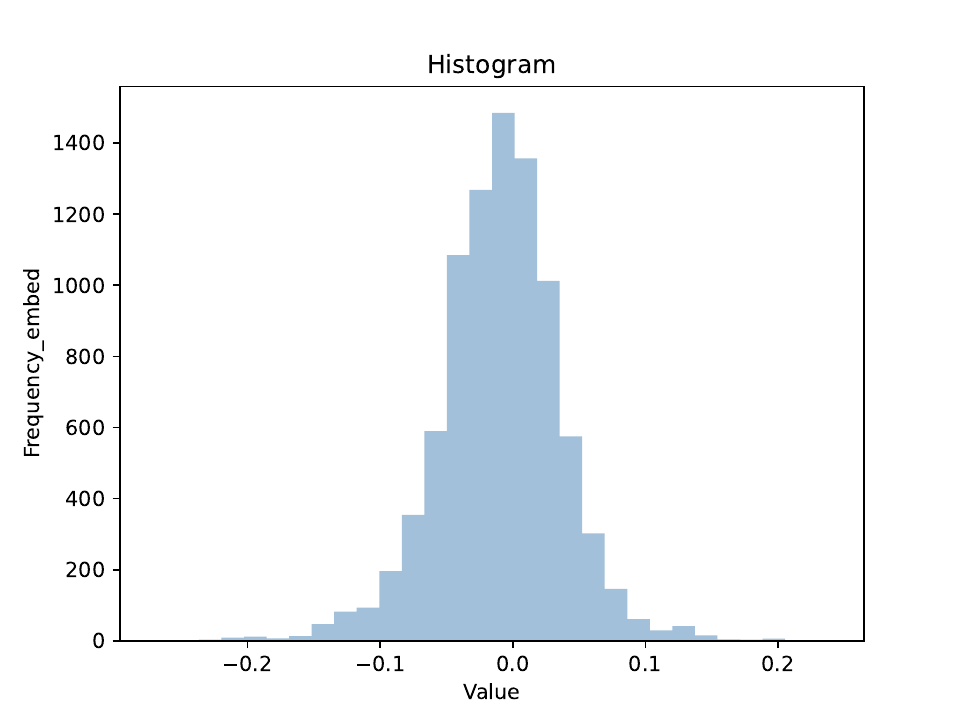} }        & \multicolumn{1}{c|}{\includegraphics[width=.2\linewidth]{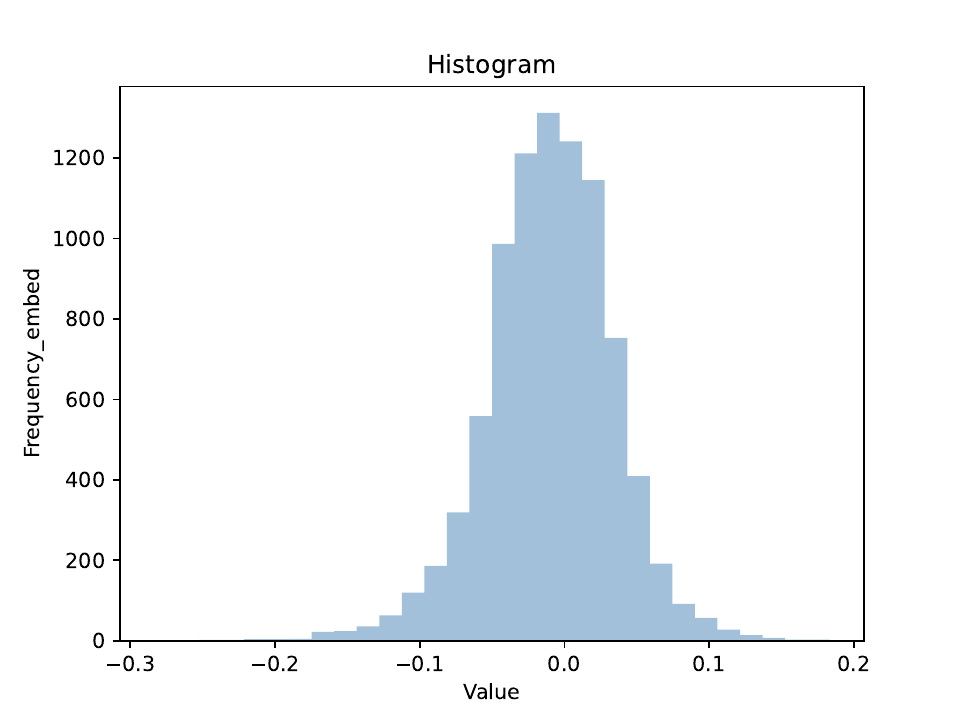} }        & \multicolumn{1}{c|}{\includegraphics[width=.2\linewidth]{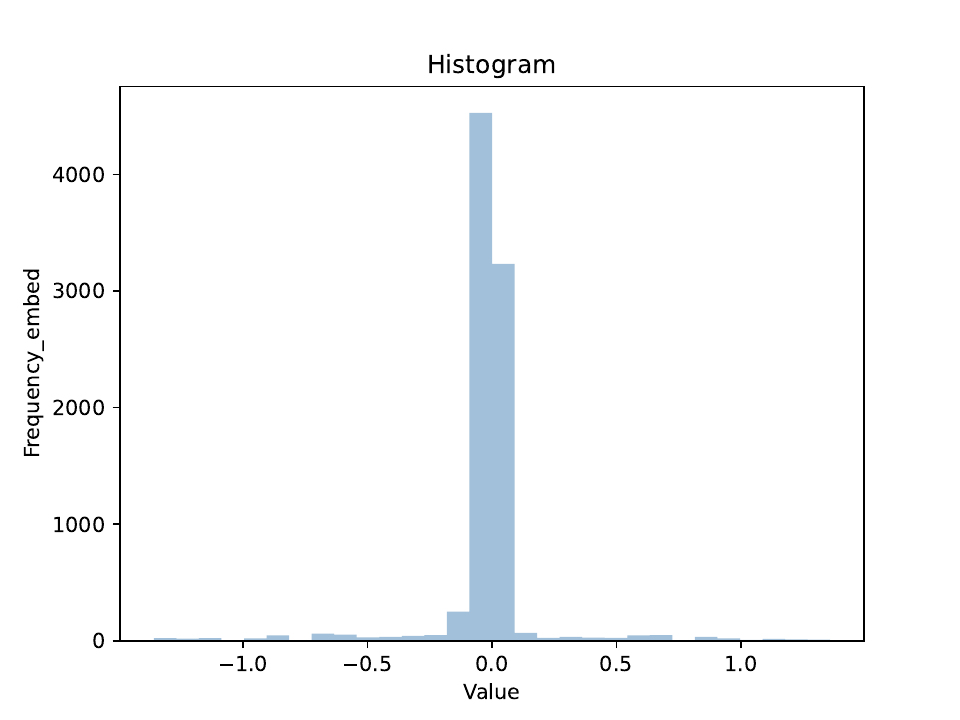} }        &           \includegraphics[width=.2\linewidth]{HistogramPicture/cnn_Histogram_embed_4_4_1_0.9.pdf}            \\ \hline \hline

\multicolumn{1}{c|}{$n$}               & \multicolumn{4}{c}{1}                                                                                             \\ \hline
\multicolumn{1}{c|}{$r$}               & \multicolumn{4}{c}{16}                                                                                            \\ \hline
\multicolumn{1}{c|}{$\Delta$}           & \multicolumn{2}{c|}{0.1}                                    & \multicolumn{2}{c}{1}                               \\ \hline
\multicolumn{1}{c|}{$\alpha$}           & \multicolumn{1}{c|}{0.9}     & \multicolumn{1}{c|}{0.5}     & \multicolumn{1}{c|}{0.9}     & 0.5                  \\ \hline
\multicolumn{1}{c|}{} & \multicolumn{1}{c|}{\includegraphics[width=.2\linewidth]{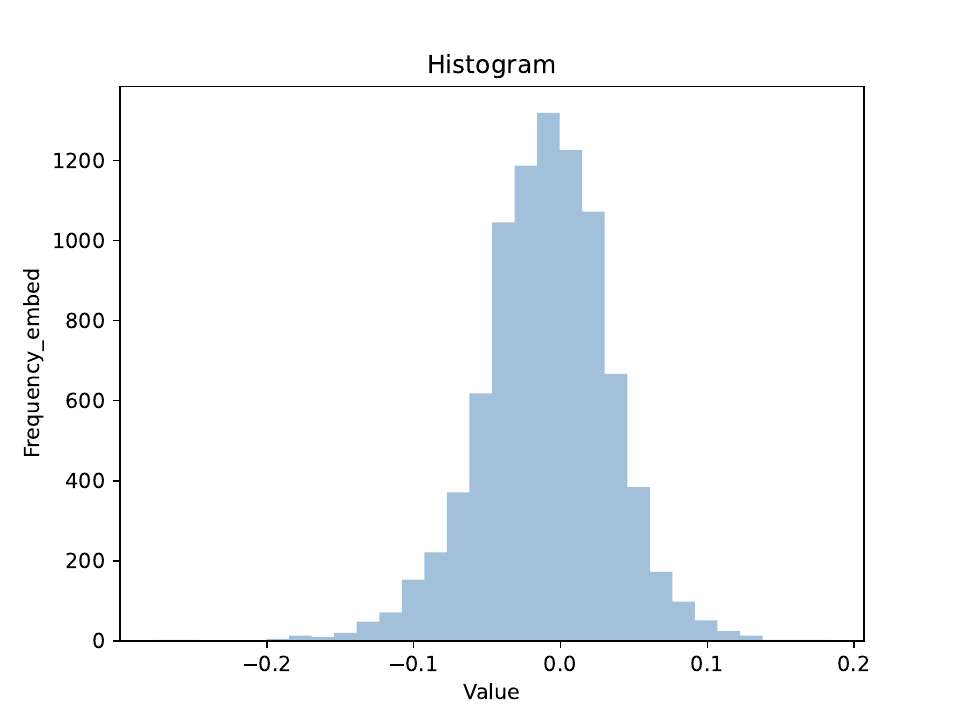} }        & \multicolumn{1}{c|}{\includegraphics[width=.2\linewidth]{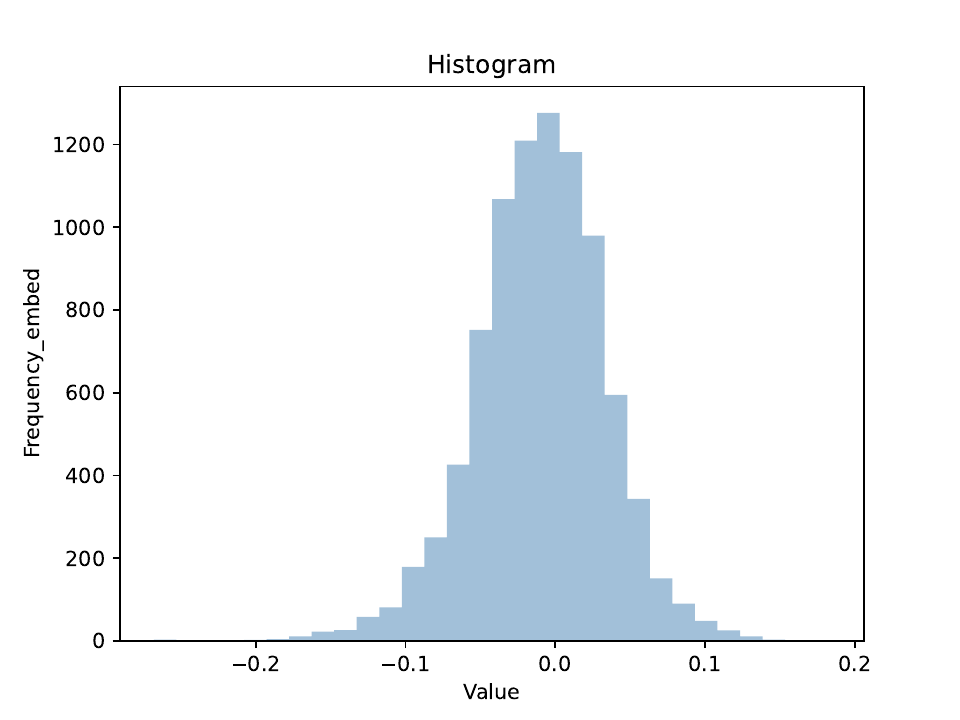} }        & \multicolumn{1}{c|}{\includegraphics[width=.2\linewidth]{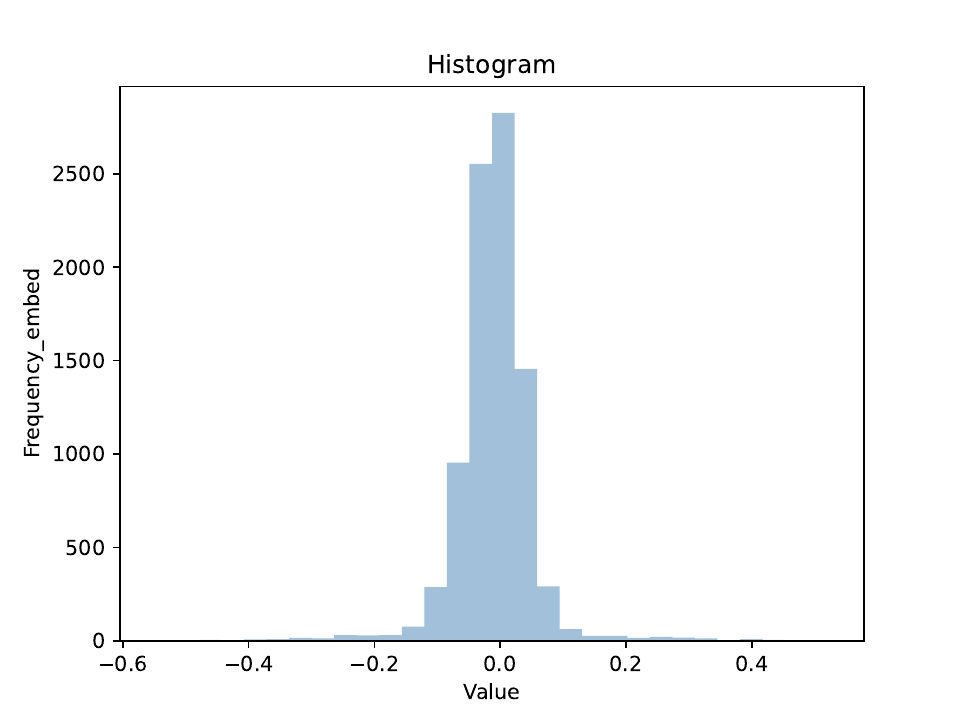} }        &          \includegraphics[width=.2\linewidth]{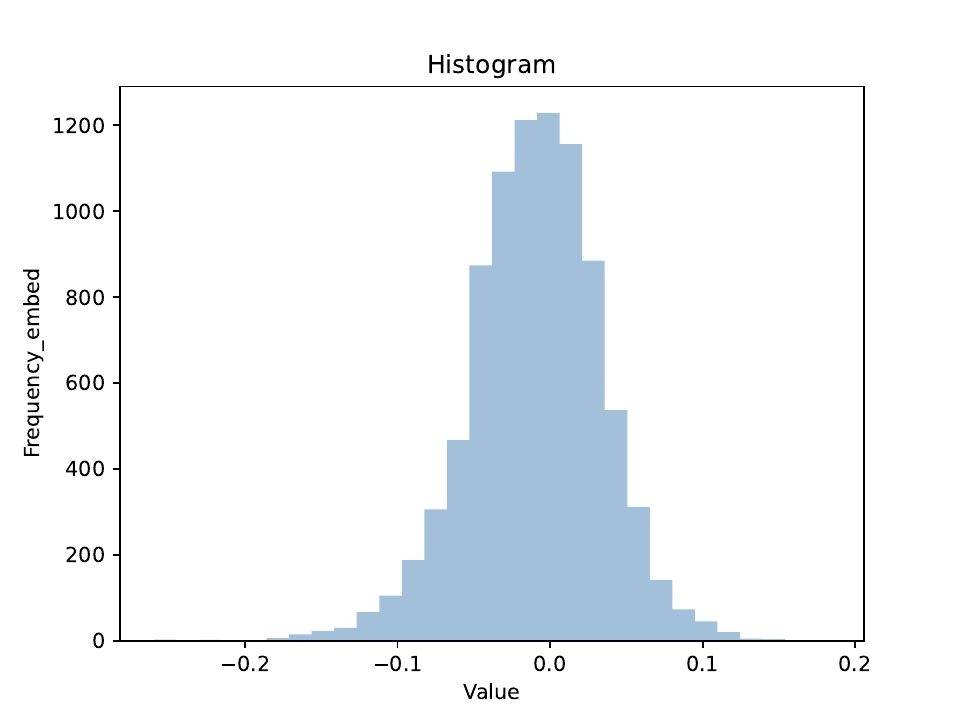}             \\ \hline  \hline

\multicolumn{1}{c|}{$n$}               & \multicolumn{4}{c}{8}                                                                                            \\ \hline
\multicolumn{1}{c|}{$r$}               & \multicolumn{4}{c}{16}                                                                                            \\ \hline
\multicolumn{1}{c|}{$\Delta$}           & \multicolumn{2}{c|}{all 0.1}                                & \multicolumn{2}{c}{all 1}                           \\ \hline
\multicolumn{1}{c|}{$\alpha$}           & \multicolumn{1}{c|}{all 0.9} & \multicolumn{1}{c|}{all 0.5} & \multicolumn{1}{c|}{all 0.9} & all 0.5              \\ \hline
\multicolumn{1}{c|}{} & \multicolumn{1}{c|}{\includegraphics[width=.2\linewidth]{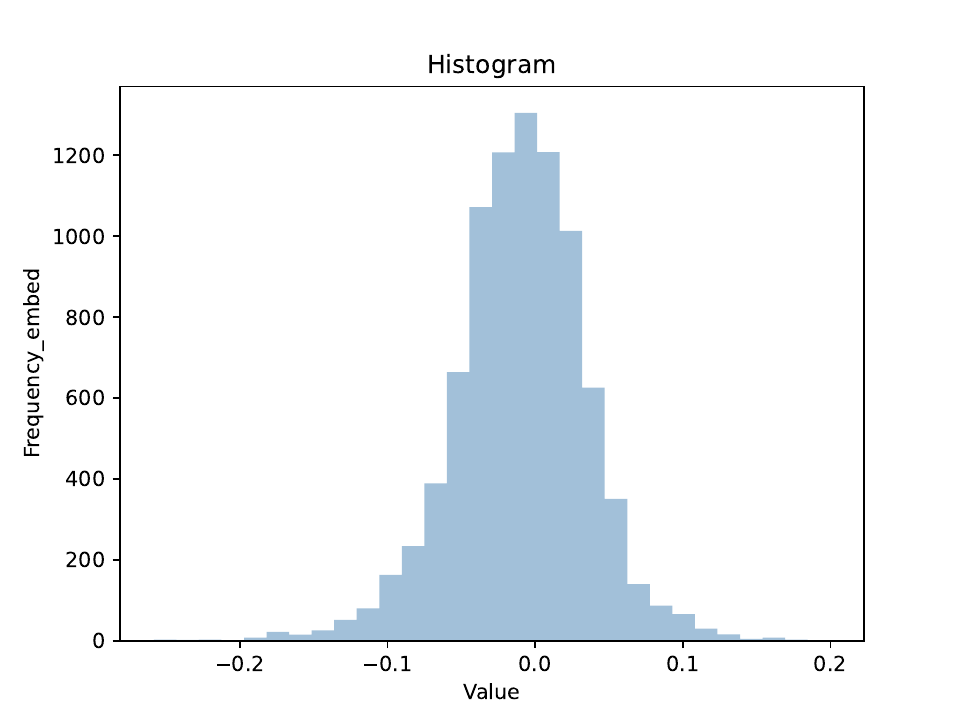} }        & \multicolumn{1}{c|}{\includegraphics[width=.2\linewidth]{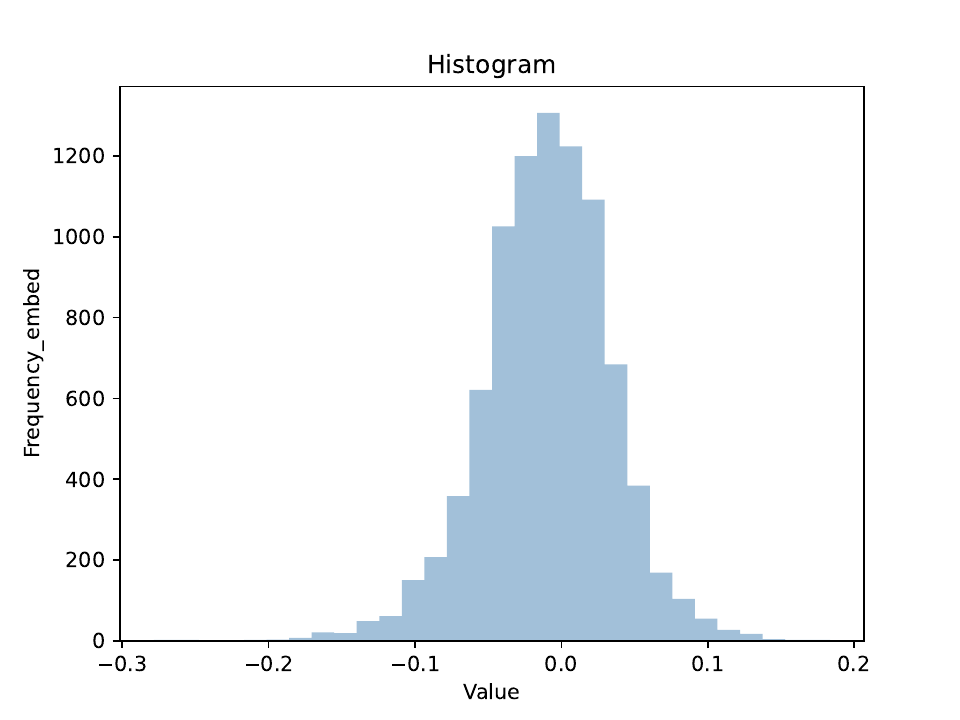 }}        & \multicolumn{1}{c|}{\includegraphics[width=.2\linewidth]{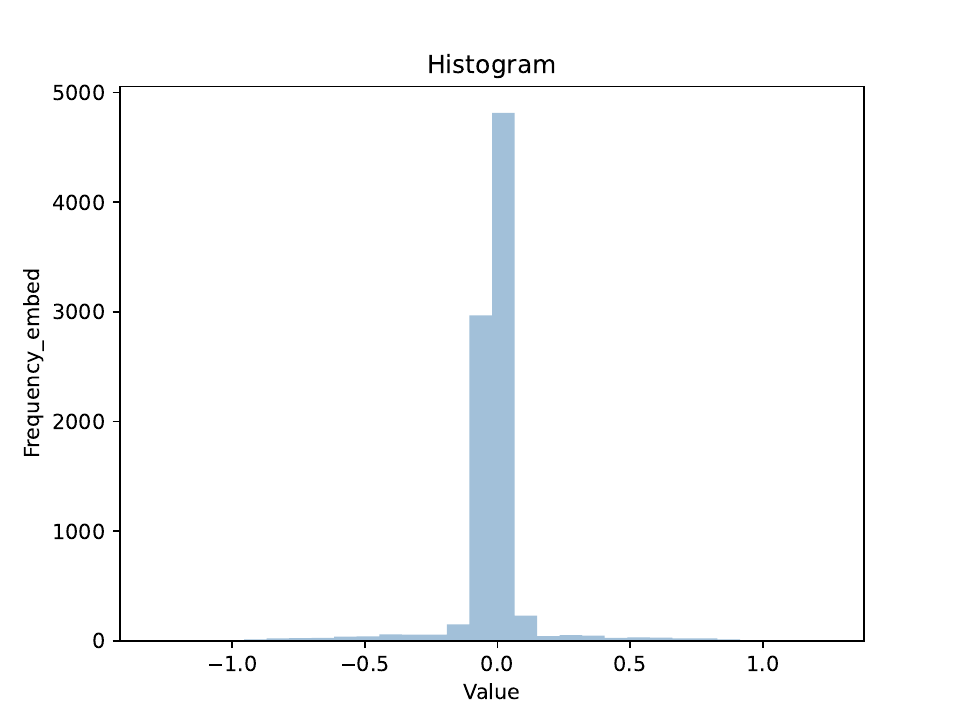} }        &       \includegraphics[width=.2\linewidth]{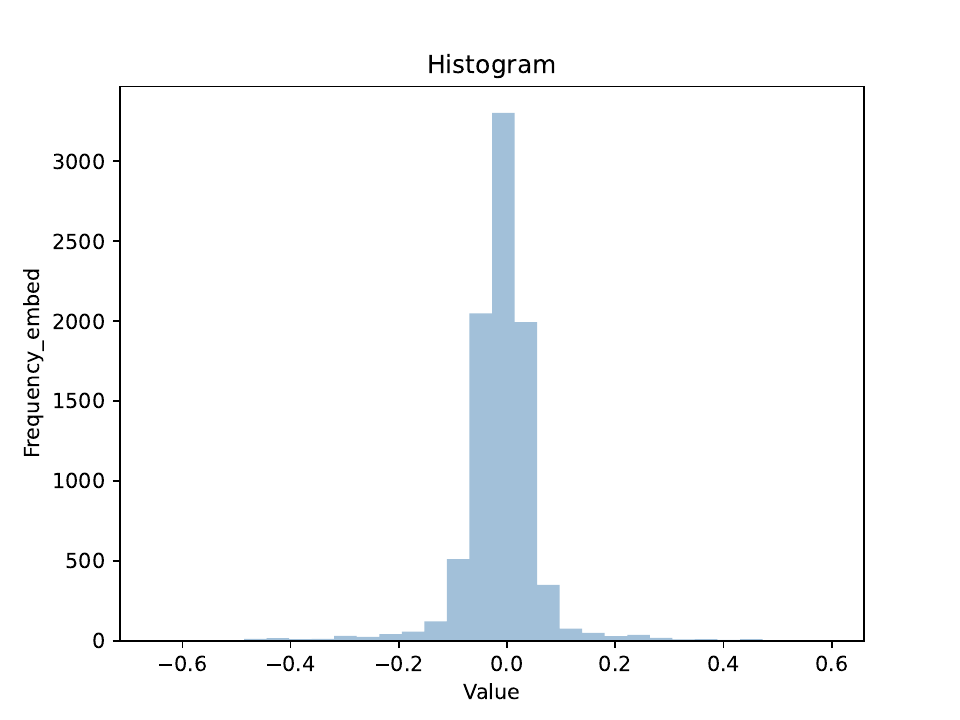}                \\ \hline
\multicolumn{1}{l}{}                 & \multicolumn{1}{l}{}         & \multicolumn{1}{l}{}         & \multicolumn{1}{l}{}         & \multicolumn{1}{l}{} \\
\multicolumn{1}{l}{}                 & \multicolumn{1}{l}{}         & \multicolumn{1}{l}{}         & \multicolumn{1}{l}{}         & \multicolumn{1}{l}{} \\
\multicolumn{1}{l}{}                 & \multicolumn{1}{l}{}         & \multicolumn{1}{l}{}         & \multicolumn{1}{l}{}         & \multicolumn{1}{l}{}
\end{tabular}
\end{table*}

\begin{table*}[]
\centering
\caption{ Histogram of watermarked MLP in different situation.}
\label{Histogram_MLP}
\begin{tabular}{ccccc}
\hline
\multicolumn{1}{c|}{$n$}               & \multicolumn{4}{c}{1}                                                                                             \\ \hline
\multicolumn{1}{c|}{$r$}               & \multicolumn{4}{c}{5}                                                                                             \\ \hline
\multicolumn{1}{c|}{$\Delta$}           & \multicolumn{2}{c|}{0.1}                                    & \multicolumn{2}{c}{1}                               \\ \hline
\multicolumn{1}{c|}{$\alpha$}           & \multicolumn{1}{c|}{0.9}     & \multicolumn{1}{c|}{0.5}     & \multicolumn{1}{c|}{0.9}     & 0.5                  \\ \hline
\multicolumn{1}{c|}{} & \multicolumn{1}{c|}{\includegraphics[width=.2\linewidth]{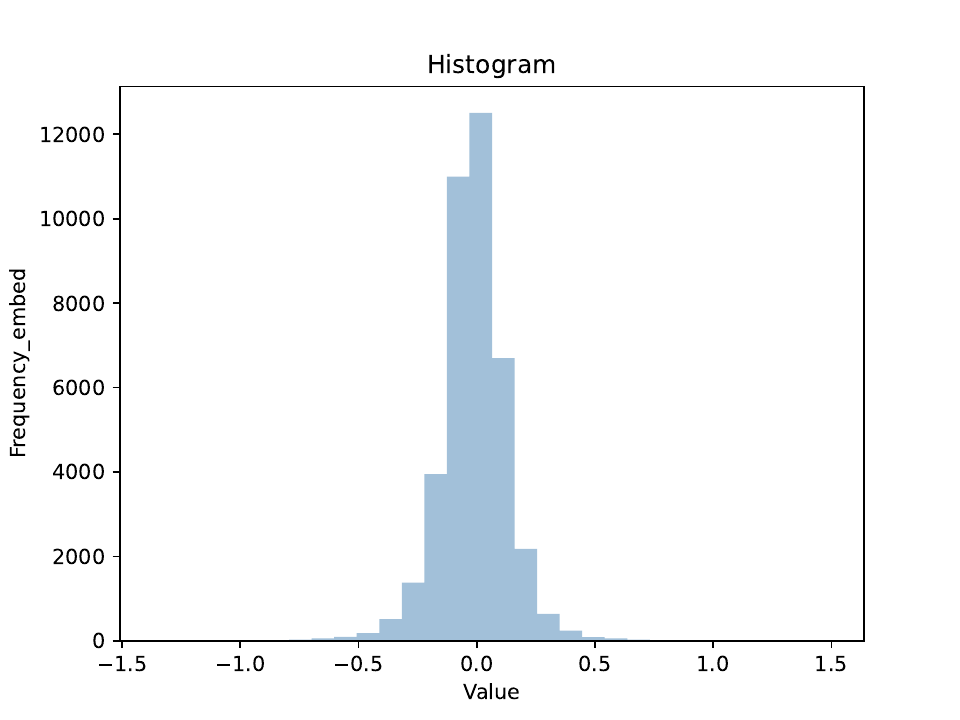} }        & \multicolumn{1}{c|}{\includegraphics[width=.2\linewidth]{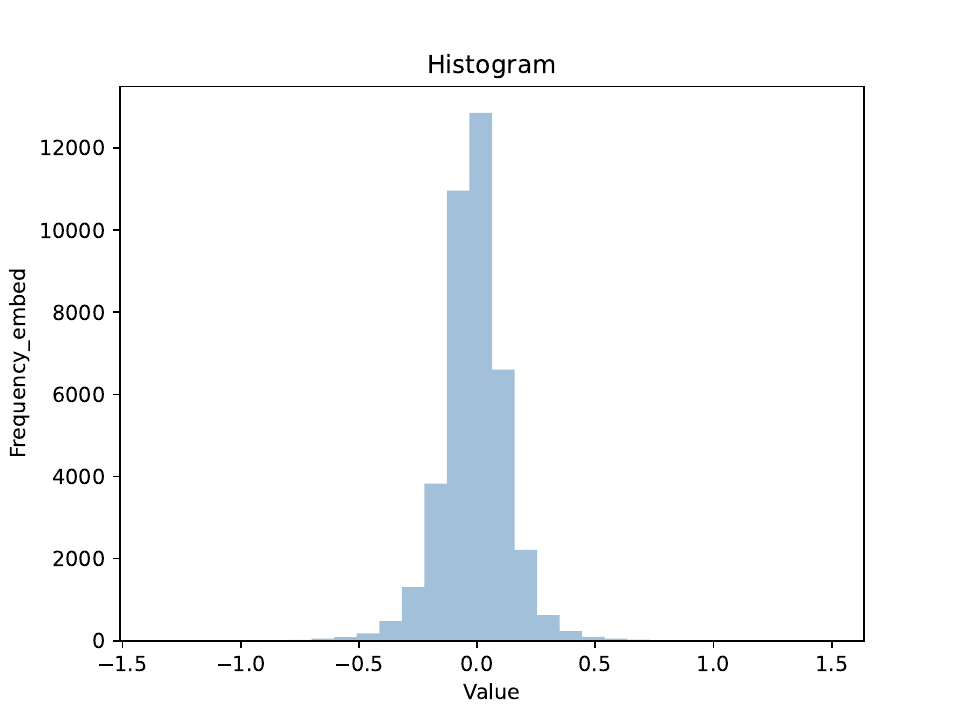} }        & \multicolumn{1}{c|}{\includegraphics[width=.2\linewidth]{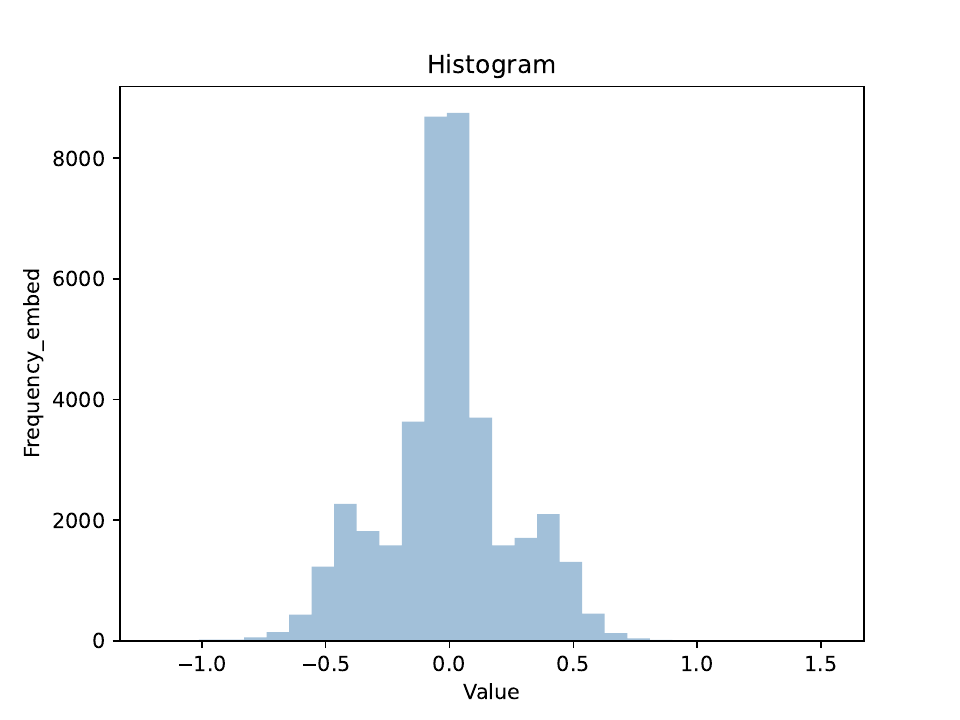} }        &         \includegraphics[width=.2\linewidth]{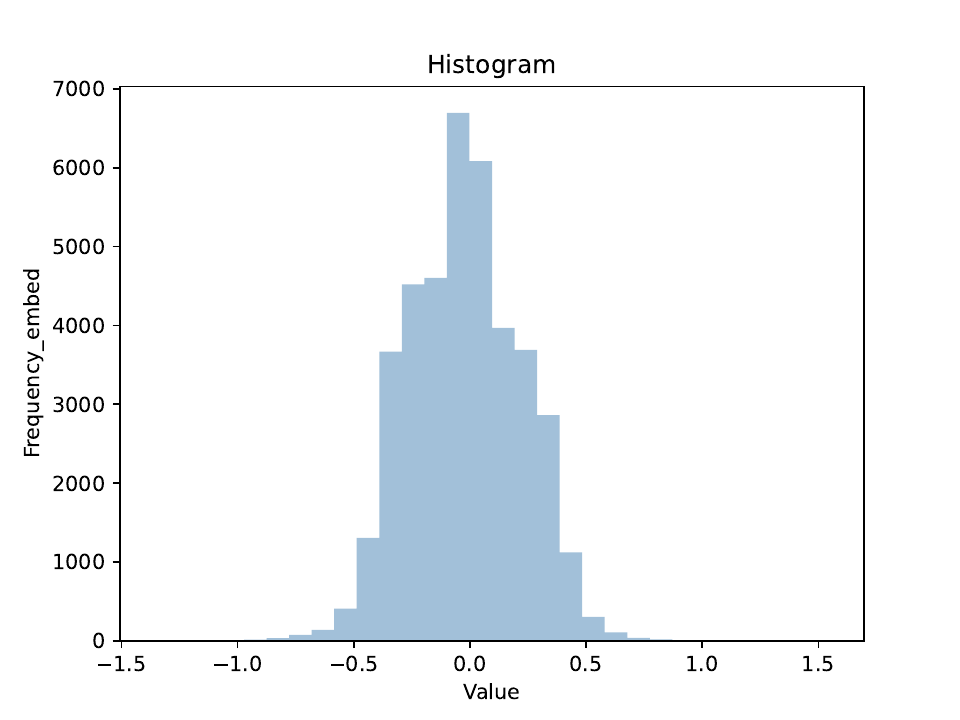}              \\ \hline \hline

\multicolumn{1}{c|}{$n$}               & \multicolumn{4}{c}{5}                                                                                             \\ \hline
\multicolumn{1}{c|}{$r$}               & \multicolumn{4}{c}{5}                                                                                             \\ \hline
\multicolumn{1}{c|}{$\Delta$}           & \multicolumn{2}{c|}{all 0.1}                                & \multicolumn{2}{c}{all 1}                           \\ \hline
\multicolumn{1}{c|}{$\alpha$}           & \multicolumn{1}{c|}{all 0.9} & \multicolumn{1}{c|}{all 0.5} & \multicolumn{1}{c|}{all 0.9} & all 0.5              \\ \hline
\multicolumn{1}{c|}{} & \multicolumn{1}{c|}{\includegraphics[width=.2\linewidth]{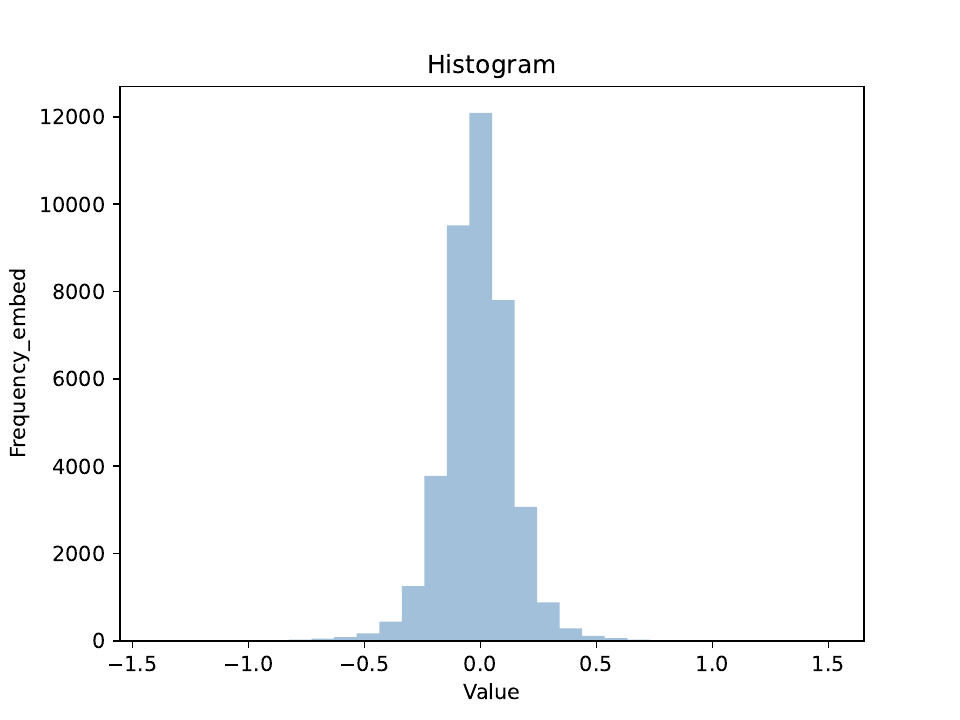} }        & \multicolumn{1}{c|}{\includegraphics[width=.2\linewidth]{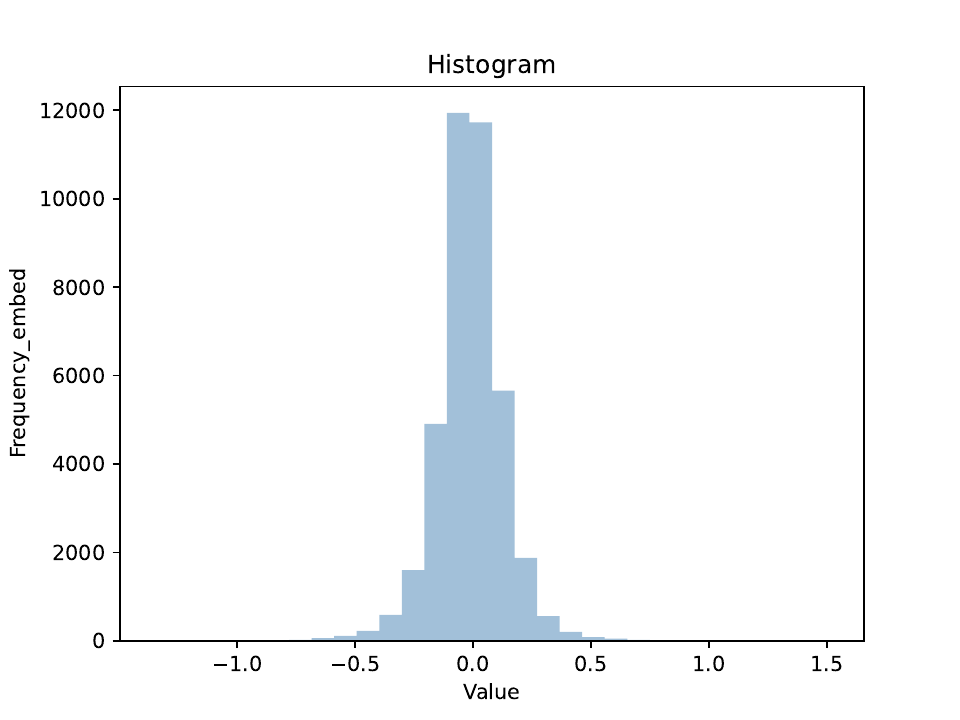} }        & \multicolumn{1}{c|}{\includegraphics[width=.2\linewidth]{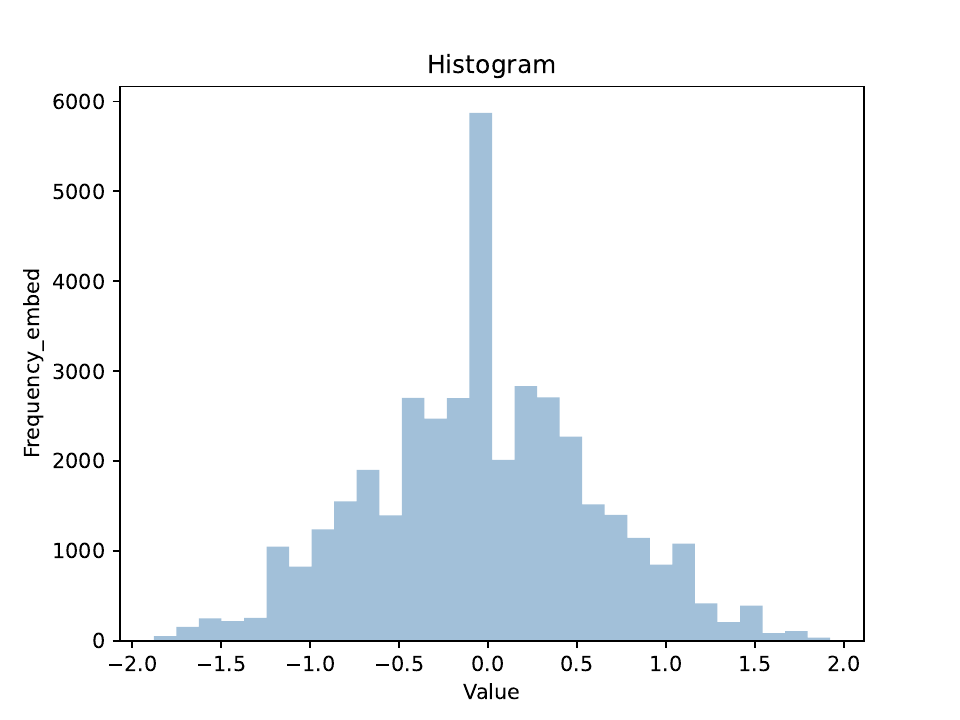} }        &           \includegraphics[width=.2\linewidth]{HistogramPicture/mlp_Histogram_embed_5_5_1_0.9.pdf}            \\ \hline \hline

\multicolumn{1}{c|}{$n$}               & \multicolumn{4}{c}{1}                                                                                             \\ \hline
\multicolumn{1}{c|}{$r$}               & \multicolumn{4}{c}{10}                                                                                            \\ \hline
\multicolumn{1}{c|}{$\Delta$}           & \multicolumn{2}{c|}{0.1}                                    & \multicolumn{2}{c}{1}                               \\ \hline
\multicolumn{1}{c|}{$\alpha$}           & \multicolumn{1}{c|}{0.9}     & \multicolumn{1}{c|}{0.5}     & \multicolumn{1}{c|}{0.9}     & 0.5                  \\ \hline
\multicolumn{1}{c|}{} & \multicolumn{1}{c|}{\includegraphics[width=.2\linewidth]{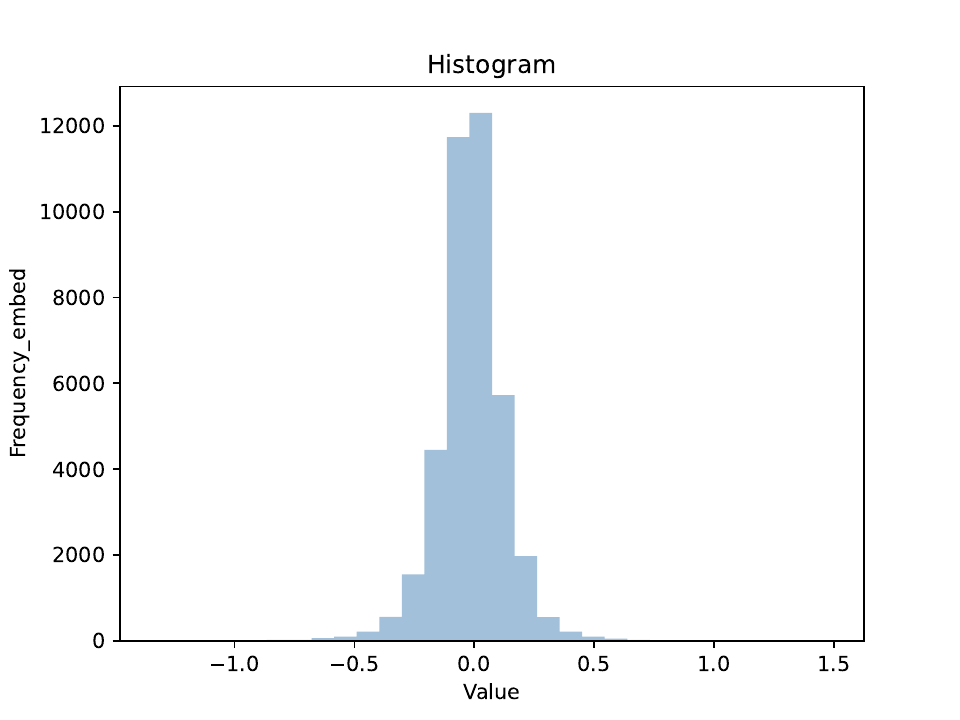} }        & \multicolumn{1}{c|}{\includegraphics[width=.2\linewidth]{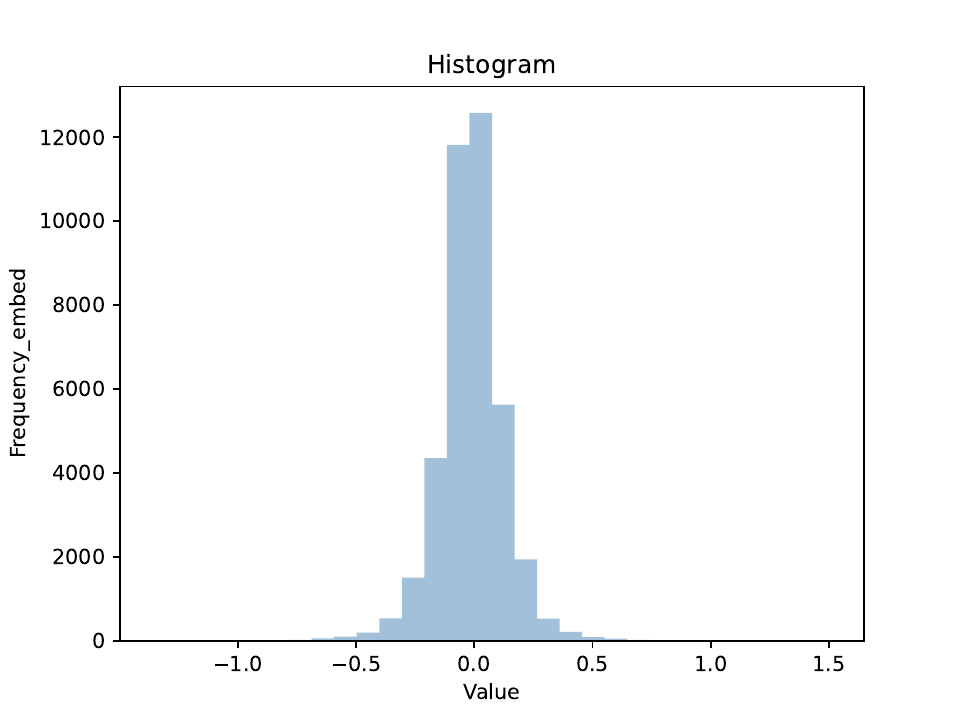} }        & \multicolumn{1}{c|}{\includegraphics[width=.2\linewidth]{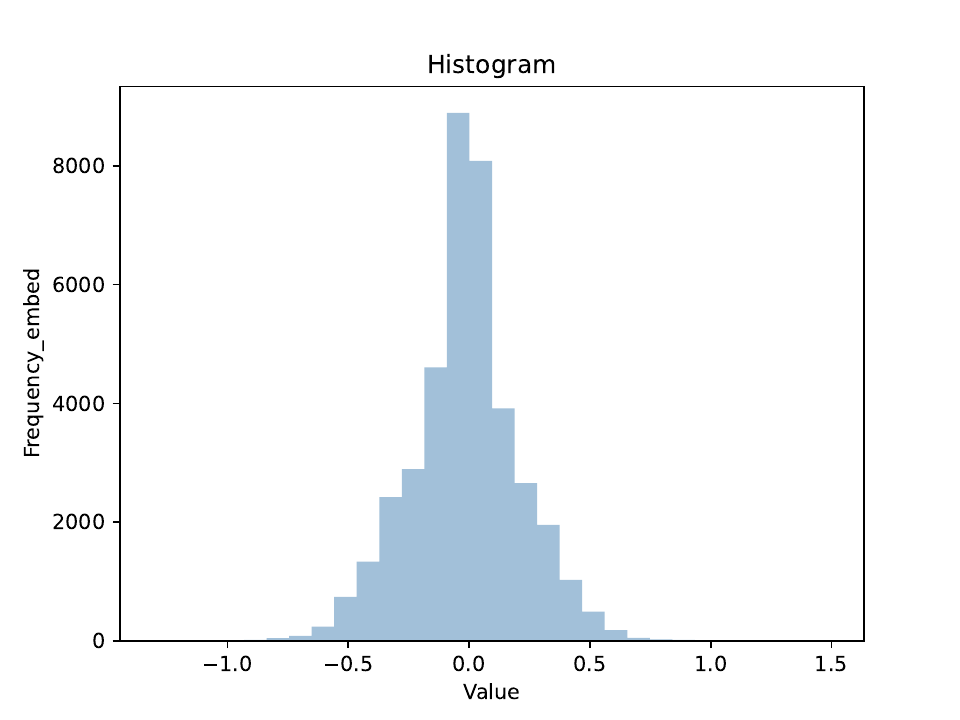} }        &          \includegraphics[width=.2\linewidth]{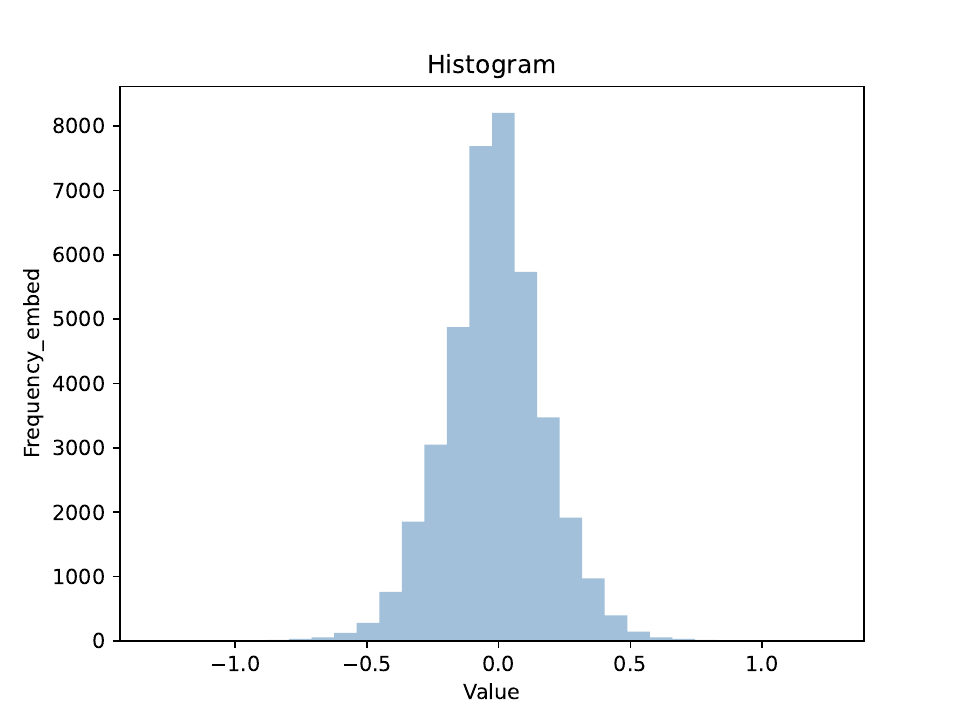}             \\ \hline  \hline

\multicolumn{1}{c|}{$n$}               & \multicolumn{4}{c}{10}                                                                                            \\ \hline
\multicolumn{1}{c|}{$r$}               & \multicolumn{4}{c}{10}                                                                                            \\ \hline
\multicolumn{1}{c|}{$\Delta$}           & \multicolumn{2}{c|}{all 0.1}                                & \multicolumn{2}{c}{all 1}                           \\ \hline
\multicolumn{1}{c|}{$\alpha$}           & \multicolumn{1}{c|}{all 0.9} & \multicolumn{1}{c|}{all 0.5} & \multicolumn{1}{c|}{all 0.9} & all 0.5              \\ \hline
\multicolumn{1}{c|}{} & \multicolumn{1}{c|}{\includegraphics[width=.2\linewidth]{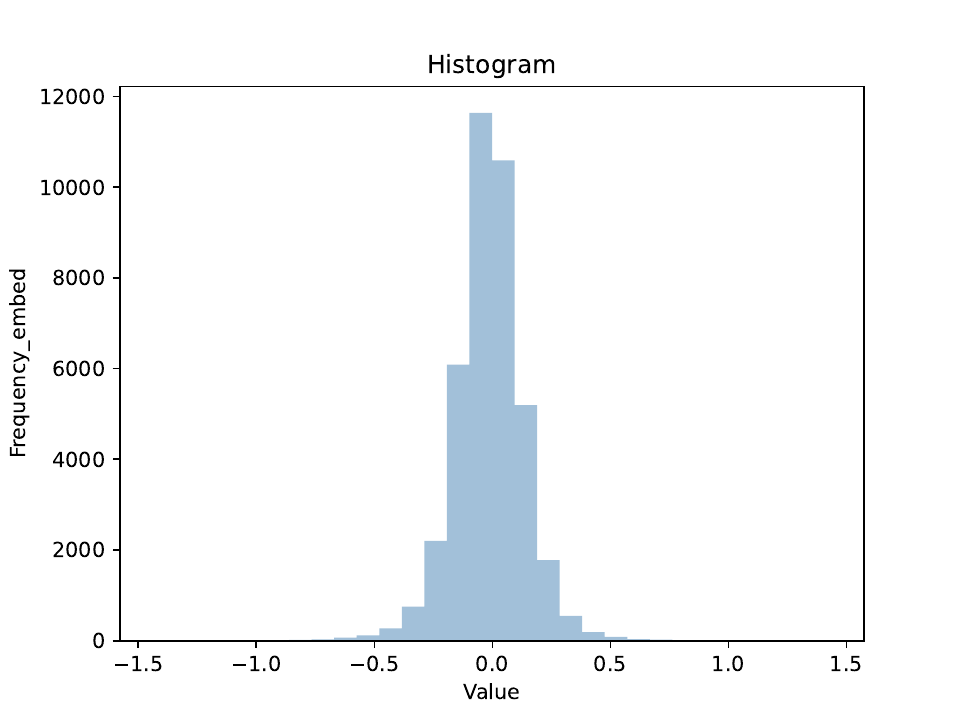} }        & \multicolumn{1}{c|}{\includegraphics[width=.2\linewidth]{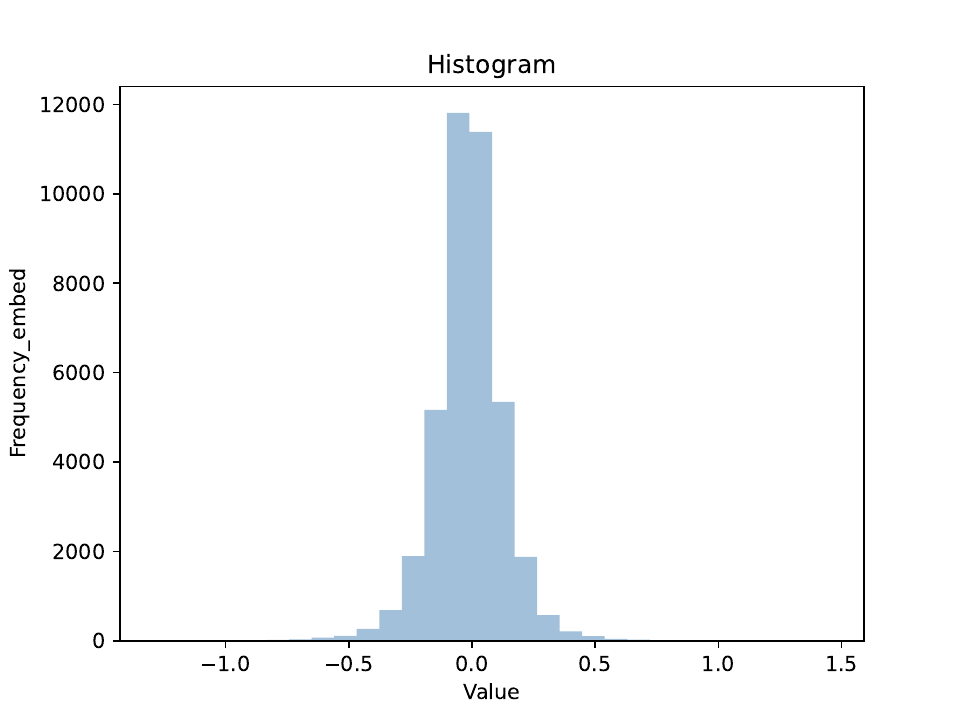 }}        & \multicolumn{1}{c|}{\includegraphics[width=.2\linewidth]{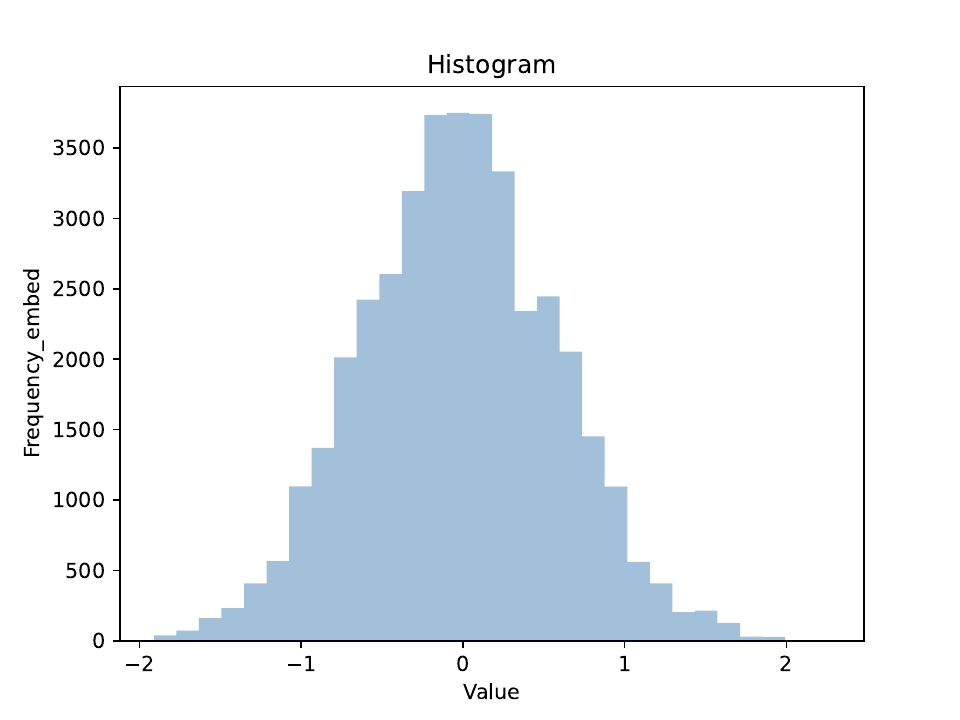} }        &       \includegraphics[width=.2\linewidth]{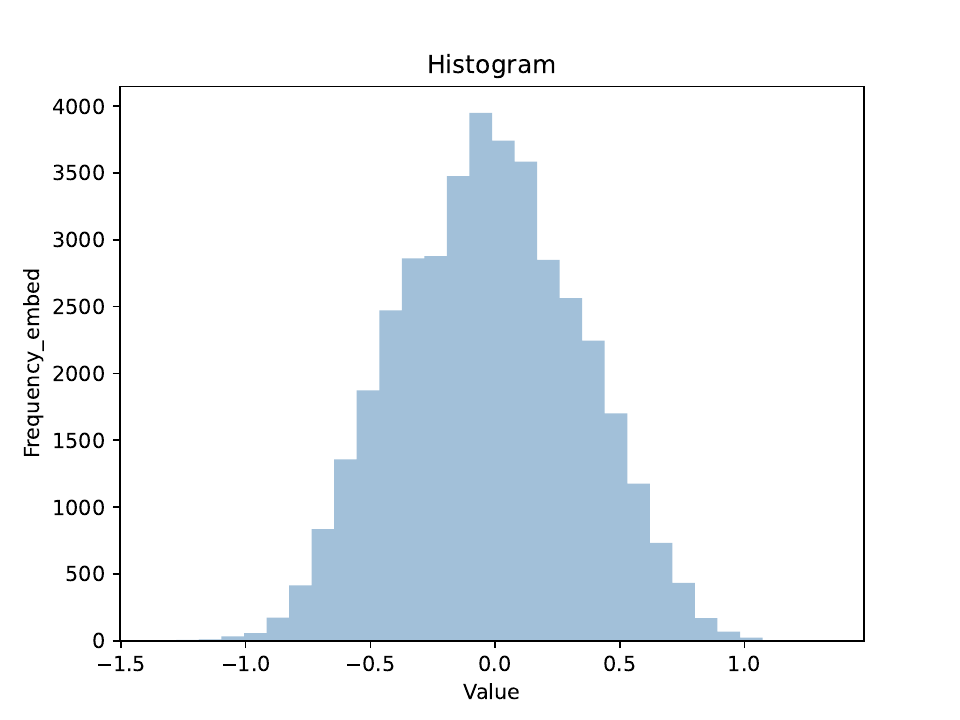}                \\ \hline
\multicolumn{1}{l}{}                 & \multicolumn{1}{l}{}         & \multicolumn{1}{l}{}         & \multicolumn{1}{l}{}         & \multicolumn{1}{l}{} \\
\multicolumn{1}{l}{}                 & \multicolumn{1}{l}{}         & \multicolumn{1}{l}{}         & \multicolumn{1}{l}{}         & \multicolumn{1}{l}{} \\
\multicolumn{1}{l}{}                 & \multicolumn{1}{l}{}         & \multicolumn{1}{l}{}         & \multicolumn{1}{l}{}         & \multicolumn{1}{l}{}
\end{tabular}
\end{table*}

\section{Conclusion}\label{Sec. conclusion}
In conclusion, this paper has introduced FedReverse, a novel multiparty reversible watermarking scheme tailored for the floating-point weights of DNNs.   
FedReverse differentiates itself by embedding watermarks from all clients into the model's weights post-training, allowing individual copyright claims and complete watermark removal from a potential buyer if he/she has obtained keys from all the clients.  FedReverse has also addressed the challenge of potential watermark conflicts among different clients through an orthogonal key generation technique, ensuring robust copyright protection. This work offers a promising ``reversible'' solution to safeguard intellectual property in the ever-expanding realm of DNN.

\bibliographystyle{IEEEtranMine}
\bibliography{reference}

\end{document}